\documentclass[10pt,journal,a4paper]{IEEEtran}\usepackage{amssymb,amsmath}
\usepackage{cite}
\usepackage{graphicx,subfigure}
\usepackage{psfrag}
\usepackage{url}
\usepackage{xcolor}
\usepackage[absolute,overlay]{textpos}
\usepackage{float}
\usepackage{lettrine}
\usepackage{indentfirst}
\usepackage{amsthm}

\newtheorem{lemma}{Lemma}
\newtheorem{theorem}{Theorem}

\usepackage{array}
\usepackage{psfrag}
\usepackage{url}
\usepackage{caption}

\makeatletter
\def\thickhline{%
  \noalign{\ifnum0=`}\fi\hrule \@height \thickarrayrulewidth \futurelet
   \reserved@a\@xthickhline}
\def\@xthickhline{\ifx\reserved@a\thickhline
               \vskip\doublerulesep
               \vskip-\thickarrayrulewidth
             \fi
      \ifnum0=`{\fi}}
\makeatother

\newlength{\thickarrayrulewidth}
\begin{document}
	
	\title{Self-energy recycling for low-power reliable networks: Half-duplex or Full-duplex?}
	\author{
        \IEEEauthorblockN{Dian Echevarría Pérez, \IEEEmembership{} %
        Onel L. Alcaraz López, \IEEEmembership{Member, IEEE},
            Hirley Alves, \IEEEmembership{Member, IEEE},\\
            Matti Latva-aho, \IEEEmembership{Senior Member, IEEE}
        } 
        
		\thanks{The authors are with the Centre for Wireless Communications (CWC), University of Oulu, Finland. \{dian.echevarriaperez,  onel.alcarazlopez, hirley.alves, matti.latva-aho\}@oulu.fi}

        \thanks{This research has been financially supported by Academy of Finland, 6Genesis Flagship (Grant no. 318927), Academy Professor (no. 307492) and EE-IoT (no. 319008).}
    }  
		\maketitle
	
	\begin{abstract}
		Self-energy recycling (sER), which allows transmit energy re-utilization, has emerged as a viable option for improving the energy efficiency (EE) in low-power Internet of Things networks. In this work, we investigate its benefits also in terms of reliability improvements and compare the performance of full-duplex (FD) and half-duplex (HD) schemes when using multi-antenna techniques in a communication system. We analyze the trade-offs when considering not only the energy spent on transmission but also the circuitry power consumption, thus making the analysis of much more practical interest. In addition to the well known spectral efficiency improvements, results show that FD also outperforms HD in terms of reliability. We show that  sER introduces not only benefits in EE matters but also some modifications on how to achieve maximum reliability fairness between uplink and downlink transmissions, which is the main goal in this work. In order to achieve this objective, we propose the use of a dynamic FD scheme where the small base station (SBS) determines the optimal allocation of antennas for transmission and reception. We show the significant improvement gains of this strategy for the system outage probability when compared to the simple HD and FD schemes.
	\end{abstract}
	
	\begin{IEEEkeywords}
    Full-duplex, ultra-reliable communications, energy harvesting, maximum ratio combining.
    \end{IEEEkeywords}
	
	\section{Introduction}
	Ultra-reliable communications (URC) aims to provide packet error rates going from $10^{-5}$ to $10^{-9}$ almost 100$\% $ of the time, and is a key operation mode in current and future wireless communication systems. Many envisioned applications need to operate with reliability levels that match those of cabled networks, e.g., vehicle-to-vehicle (V2V) communications require error free operations 99.999$\%$ of the time in the basic version of the service, teleprotection in smart grids 99.999$\%$, factory automation 99.9999999$\%$ and remote surgery 99.99999999$\%$ \cite{popovski2019wireless,6815890,6881174}.\newline 
    \indent Some cutting edge technologies have been linked to URC\footnote{We are aware of the extreme mode known as ultra-reliable low-latency communications (URLLC) \cite{popovski2019wireless}, and other hybrid URC models foreseen for 6G \cite{Mahmood.2020}. However, herein we focus on native URC that has broader application.}  in recent years. That is the case of radio frequency (RF) energy harvesting (EH), which has become a viable option for enhancing the energy efficiency (EE) in low-power Internet of Things (IoT) networks by using the energy carried in RF waves, thus supporting longer devices' lifetimes\cite{10,lopez2021massive}. Some works have focused on this direction, for instance, the authors in \cite{lopez2017ultrareliable} analyzed the delay and error probability in a scenario with wireless energy transfer (WET) in the downlink (DL) and wireless information transfer (WIT) in the uplink (UL) under strict latency and reliability requirements. Also, a short-packet cooperative scheme enabling URC is studied in [11], where the relays are powered via WET. Therein, the error probability is investigated under perfect and imperfect channel state information (CSI). 
    Meanwhile, full-duplex (FD) technology has been widely studied in recent years because of its potential to promote reliability and latency enhancements.
    FD also allows an improvement in spectral efficiency (SE) since transmission and reception occur simultaneously using the same frequencies with the help of self-interference cancellation (SIC) techniques \cite{1, 2}. For example,  the authors in \cite{gu2017ultra} compared the performance of FD and half-duplex (HD) relaying for URC. 
    
    Recently, a novel RF-EH concept termed as self-energy recycling (sER) has also captured the attention of both academia and industry. sER is a technique that allows to re-utilize part of the transmitted energy by means of an ``energy-loop'' \cite{3,lopez2020full}.  
    FD and sER techniques can be combined allowing a higher reuse of the transmitted energy since the recycling process occurs uninterruptedly while the devices are transmitting \cite{3,4,7444859,7519055}. The benefits of this combination in terms of SE, EE, among others, have been investigated recently, e.g., \cite{13,3,11,shaikh2020energy}. However, the impact of sER on the reliability performance of either HD or FD configuration has not been investigated in the literature. For instance, authors in \cite{13} studied the benefits of performing sER in an FD small base station (SBS) in terms of EE but they did not consider the benefits regarding reliability. An FD wireless-powered relay protocol with sER is proposed in \cite{13}, while in \cite{11}, the focus is on sER in an FD relaying multiple-input multiple-output (MIMO) orthogonal
    frequency division multiplexing system. Rather than network reliability, the authors maximized the SE and derived the optimal power allocation per antenna. Authors in \cite{shaikh2020energy} investigated the impact of sER on the EE and SE in an FD multiuser MIMO system reserving a few antennas for sER from the total. However, they derived expressions for the outage probability in the DL channel only and did not focus on reliability improvements. \newline
    There are other related works where the benefits of sER were not considered, for instance, the work in \cite{nguyen2020transmit} studied a point-to-point bidirectional spatial modulation FD MIMO system with transmit antenna selection in Rayleigh fading channels. Expressions for the outage probabilities were derived and a dynamic/flexible arrangement of transmit and receive antennas was proposed.
    Authors in \cite{zhai2018accumulate} evaluated an accumulate-then-transmit framework for multiuser scheduling in an FD wireless-powered system with multiple IoT devices powered via EH and one FD hybrid access point (HAP). The FD HAP operates with limited number of antennas and does not exploit the benefits of sER. Similarly, in \cite{zhao2019hybrid}, the authors analyze a simultaneous wireless information and power transfer system with FD
IoT nodes. The total transmit power is minimized with a joint optimization of the hybrid precoder, decoding rule, and power splitting ratio, all this with fixed number of antennas for transmission and reception and without exploiting sER. Meanwhile, authors in \cite{guo2019performance} minimized the system outage probability in the presence of an FD decode-and forward EH relay by properly tuning the power splitting factor. Therein, sER was not available and devices were assumed to be equipped with single antennas. The design of a low complexity FD transceiver with antenna subset selection in a relay network was proposed in \cite{xia2014low}. In that work, the authors determined the number of antennas required for transmission and reception at the relay node for maximum diversity gain. However, the focus was not on reliability and again sER was not exploited.
It is worth mentioning that most of these works assume that the energy is only consumed on transmission, thus ignoring the energy consumed in the RF circuits. \newline 
\indent None of previous  works has considered the impact of the use of the self-recycled energy on the system reliability, while the number of transmit and receive antennas always remains fixed.  Aiming at filling this research gap, herein\footnote{Notice that our work here  summarizes the main outcomes from Dian's Master thesis in \cite{echevarria_perez_2021}, but extending the analysis towards a reliability-centric approach  and evaluating the  impact of the number of EH antennas.} we focus on the outage probability analysis of an antenna-dynamic FD system with sER. Our main contributions are three-fold:
 \begin{itemize}
     \item we derive  closed-form expressions for the outage probabilities of HD and FD schemes under the influence of sER. We show numerically that FD usually outperforms HD in terms of reliability when considering a fixed number of transmit and receive antennas;
     \item we propose a FD scheme that dynamically adopts the configuration of transmit and receive antennas that minimizes the maximum outage probability of UL and DL channels assuming the use of state-of-the-art transceivers \cite{culbertson2003full}. The proposed FD strategy outperforms traditional HD and FD schemes in terms of outage probability;
     \item we discuss the performance gains from sER, as well as its influence on the optimal system configuration. For instance, we show that for 6 EH antennas the optimal configuration changes to 5 transmit and 11 receive antennas with respect to the setup without sER (6 transmit and 10 receive antennas). 
 \end{itemize}
 
 The work is structured as follows. In Section \ref{section_system}, we briefly describe the system model and main assumptions. In Section~\ref{section_HD}, we analyze the HD scenario and derive the expressions for the outage probability and harvested energy, while we focus on the FD scenario in Section  \ref{section_FD}. Section \ref{results} discusses numerical results, and finally, Section \ref{conclusion} concludes the paper.\\
	\\\textbf{Notation}: Boldface lowercase letters denote vectors whereas boldface uppercase letters denote matrices. For instance $\textbf{x}=\{\textbf{x}_i\}$, where $x_i$ is the $i$-th element of vector \textbf{x}; $||\cdot||$ is the Euclidean norm of a vector. The superscript $(\cdot)^*$ denotes complex conjugate of a vector, and $(\cdot)^H$ is the conjugate transpose. $E\{\cdot\}$ denotes expectation. $\Gamma(a)$ is the gamma function, $\Gamma(a,b)$ is the  upper incomplete gamma function with parameters $a$ and $b$ \cite[Eq.(8.2.1-8.2.2)]{olver2010nist}, and $\beta(a,b)$ is the beta function. $_1F_1(\cdot)$ is the confluent hypergeometric function of first kind\cite[Eq.(13.2.2)]{olver2010nist} and $_2F_1(\cdot)$ is the Gauss hypergeometric function\cite[Eq.(15.2.1)]{olver2010nist}. The acronyms and main symbols used along the paper are summarized in Table~\ref{table_0}.

	\begin{table}[t!]
    \centering
    \caption{Acronyms and Symbols}
    \label{table_0}
    \begin{tabular}{l  l}
        \thickhline
            \textbf{Acronym} & \textbf{Definition} \\
        \hline
            AWGN & additive white Gaussian noise\\
            CDF & cumulative distribution function\\
            CSI & channel state information \\
            GPD & generalized Pareto distribution\\
            HD, FD & half/full duplex \\
            IoT & Internet of Things \\
            MIMO & multiple-input multiple-output\\
            MRT/MRC & maximum ratio transmission/combining \\
            PDF & probability density function\\
            SBS & small base station\\
            sER & self-energy recycling \\
            SIC & self-interference cancellation \\
            SINR & signal-to-interference-plus-noise ratio\\
            SNR & signal-to-noise ratio \\
            UL/DL  & uplink/downlink \\
            URC & ultra-reliable communications \\
            WET/WIT & wireless energy/information transfer \\
            \thickhline
            \thickhline
            \textbf{Symbol} & \textbf{Definition} \\
            \hline
            $\varphi_\mathrm{td},\varphi_\mathrm{ur},\varphi_\mathrm{ud}$ & path gains in links $TD$, $UR$, $UD$ \\
            $\varphi_{g},\varphi_\mathrm{SI}$ & path gain in $\mathrm{sER}$ and self-interference links\\
            $\gamma_{d}/ \gamma_{sbs}$ & SNR at D/SBS \\
            $\textbf{w}_\mathrm{td}$, $\textbf{w}_\mathrm{ur}$  & beamforming vectors in TD and UR\\
            $\textbf{h}_\mathrm{td}$, $\textbf{h}_ \mathrm{ur}$  & vectors of channel coefficients in TD and UR \\
            $P_\mathrm{EH}$ & power drawn from sER\\
            $P_\mathrm{G}$ & power provided by the source\\
            $P_\mathrm{RF}$ & $P_\mathrm{G}$ minus power consumed by active elements\\
            $h_\mathrm{ud}$ & channel coefficient in the link UD\\
            $\zeta$ & SIC coefficient\\
            $\eta$ & conversion efficiency in EH\\
            $\tau$ & fraction of time where the SBS transmits\\
            $T$ & time duration block\\
            $M/N/P$ & number of transmit(receive)/receive/EH antennas\\
            $Q$ & total number of RF chains \\
        \thickhline
    \end{tabular}
\end{table}

	\section{System Model and Assumptions}\label{section_system}
	We consider a wireless communication system with an SBS serving one DL and one UL single-antenna device operating at a constant rate ($r_d$ and $r_{sbs}$, respectively) under quasi-static Rayleigh fading, i.e., the channel remains constant during a transmission block and changes independently from block to block. The SBS is assumed to be powered by a regular source that can provide $P_\mathrm{G}$ power units at most. In addition, part of the transmitted energy is recovered at the SBS via sER and it is used to assist the communication.  FD and HD duplex configurations are analyzed, and in both scenarios, diversity techniques such as maximum ratio transmission (MRT) and maximum ratio combining (MRC)\footnote{We assume MRT and MRC because they can be implemented easily in practice if full CSI is available. Also, some works have considered them, for instance \cite{3,wu2017robust}, while other schemes have been studied in the literature together with WET, such as transmit and receive zero-forcing \cite{mohammadi2015full}.} are used at the SBS. Finally, the total number of  RF transceivers is denoted by $Q$, while full CSI is assumed at the SBS and $D$ in order to perform coherent decoding methods, which is a common assumption in these settings \cite{lopez2017ultrareliable,lopez2017wireless}.

	\section{Half-Duplex}\label{section_HD}
	 First, we consider an HD SBS equipped with $M=Q \geq2$ antennas used for both transmission  ($ T_x $) and reception ($R_x$), and $P$ antennas for sER, thus $P = 0$ means that sER cannot be exploited. The scenario is shown in Fig. \ref{fig:HD}, where TD and UR represent the links $T_x\rightarrow D$ and $U\rightarrow R_x$, respectively. The communication is divided into two phases as discussed next. 	 %
	 \subsection{HD Phases}
	 In the first part with duration $\tau T$ with $\tau \in (0,1)$, the SBS transmits the signal intended to $D$  while $U$ keeps silent. sER is also carried out in this phase. MRT with beamforming vector for information transmission $\textbf{w}_{\scalebox{.7} {td}}=\textbf{h}_{\scalebox{.7} {td}}/\left\lVert \textbf{h}_{\scalebox{.7} {td}}\right\rVert$ is used\cite{6}, where $\textbf{h}_{\scalebox{.7} {td}}=[h_1\;\!\!,h_2\;\!\!,...\;\!\!,h_{\scalebox{.7} {M}}]^T$ is a column vector containing the zero-mean unit-variance complex Gaussian channel coefficients between the $M$ transmit antennas and $D$. 
	\begin{figure}[t!]
	    \centering
	    \includegraphics[height=2.1in,width=3.2in]{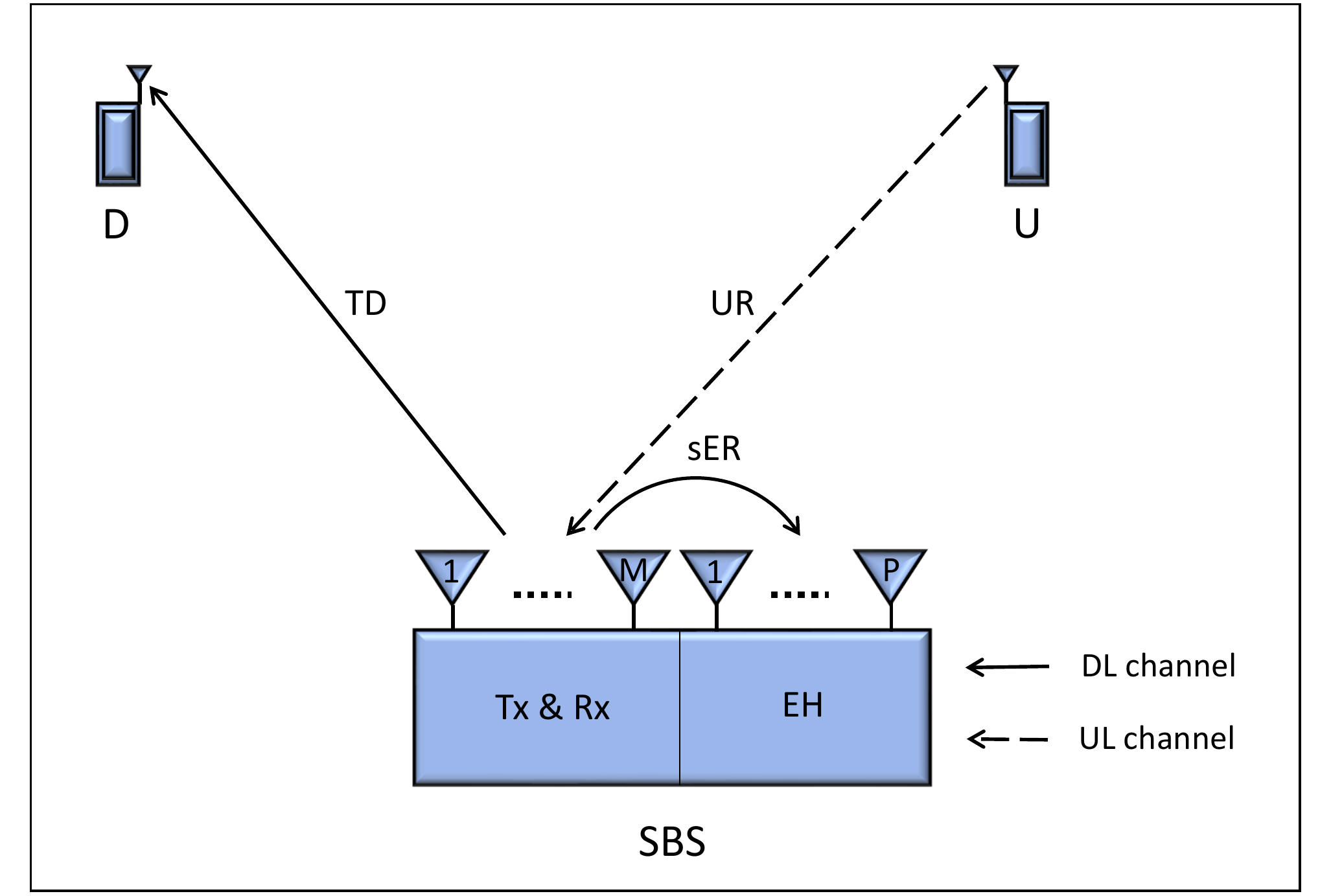}
	    \caption{Communication system operating in HD mode. The SBS communicates with $D$ and $U$ through the DL and UL channels, respectively. It is also equipped with $Q$ RF chains and P antennas for EH.}
	    \label{fig:HD}
	\end{figure}
    Hence, the signal received at $D$ can be modelled as follows
	\begin{align}
    \text{y}_{\scalebox{.7} {d}}=&\sqrt{\varphi_{\scalebox{.7} {td}}(P_{\scalebox{.7} {EH}}+P_{\scalebox{.7} {RF}})}\textbf{w}_{\scalebox{.7} {td}}^H\textbf{h}_{\scalebox{.7} {td}}s_s+n_{\scalebox{.7} {d}}\nonumber\\
    =&\sqrt{\varphi_{\scalebox{.7} {td}}(P_{\scalebox{.7} {EH}}+P_{\scalebox{.7} {RF}})}||\textbf{h}_{\scalebox{.7} {td}}||s_s+n_{\scalebox{.7} {d}},
    \end{align}
	where $P_{\scalebox{.7} {EH}}$ and $P_{\scalebox{.7}{RF}}$ represent the power drawn from the sER process and the fixed power  available after active elements consumption at the SBS, respectively. Note that for simplicity most of works assume the ideal case where $P_{\scalebox{.7}{RF}} = P_\mathrm{G}$, which means no power consumption in active elements and therefore, all power is used for transmission. We discuss this in detail at the end of this section. 
	The parameter $s_s$ represents the baseband transmitted symbol such that $E\{|s_s|^2\}$=1, $n_{\scalebox{.7} {d}}$ is the additive white Gaussian noise (AWGN) with power $\sigma^2$, and $\varphi_{\scalebox{.7} {td}}$ represents the path gain (path gain = path loss$^{-1}$) in the link TD. Therefore, the signal-to-noise ratio (SNR) at $D$ is given by
	\begin{align}
    \gamma_{\scalebox{.7} {d}}=\frac{\varphi_{\scalebox{.7} {td}}(P_{\scalebox{.7} {EH}}+P_{\scalebox{.7} {RF}})||\textbf{h}_{\scalebox{.7} {td}}||^2}{\sigma^2},
    \label{snr_1}
    \end{align}
   where $||\textbf{h}_{\scalebox{.7} {td}}||^2$ follows a Gamma distribution since it is the sum of $M$ unit-mean exponential random variables.\newline
	\indent In the second phase of duration $(1-\tau)T$,  $U$ starts transmitting the information-bearing to the SBS. This signal does not impact the EH process because its power is too low and it is not comparable to the level of self-interference. MRC is performed at this point with combining vector $\textbf{w}_{\scalebox{.7} {ur}}=\textbf{h}_{\scalebox{.7} {ur}}/\left\lVert \textbf{h}_{\scalebox{.7} {ur}}\right\rVert$, where $\textbf{h}_{\scalebox{.7} {ur}}\!=\![h_1\;\!\!,h_2\;\!\!,...,\!\!\;h_{\scalebox{.7} {M}}]^T$ is a column vector containing the zero-mean unit-variance complex Gaussian channel coefficients between $U$ and the $M$ antennas of the SBS. The combined signal is then given by
\begin{equation}
\text{y}_{\scalebox{.7} {sbs}}=\sqrt{\varphi_{\scalebox{.7} {ur}}P_{\scalebox{.7} {u}}}\textbf{w}_{\scalebox{.7} {ur}}^H\textbf{h}_{\scalebox{.7} {ur}}s_{\scalebox{.7} {u}}+\textbf{w}_{\scalebox{.7} {ur}}^H\textbf{n}_{\scalebox{.7} {s}},
\end{equation}
where $\varphi_{\scalebox{.7} {ur}}$ represents the link's path gain, $s_{\scalebox{.7} {u}}$ is the baseband transmitted symbol such that $E\{|s_{\scalebox{.7} {u}}|^2\}=1$, $P_{\scalebox{.7} {u}}$ is the power transmitted by U, and $\textbf{n}_{\scalebox{.7} {s}}$  is the AWGN vector such that each element is assumed uncorrelated with same power $\sigma^2$. 
Therefore, the SNR at the SBS is given by \newline 
\begin{align}
\label{outage_sbs}
    \gamma_{\scalebox{.7} {sbs}}&= \frac{|\sqrt{\varphi_{\scalebox{.7} {ur}}P_{\scalebox{.7} {u}}}\textbf{w}_{\scalebox{.7} {ur}}^H\textbf{h}_{\scalebox{.7} {ur}}s_{\scalebox{.7} {u}}|^2}{E\{|\textbf{w}_{\scalebox{.7} {ur}}^H\textbf{n}_{\scalebox{.7} {s}}|^2\}}= \frac{\varphi_{\scalebox{.7} {ur}}P_{\scalebox{.7} {u}}|\textbf{w}_{\scalebox{.7} {ur}}^H\textbf{h}_{\scalebox{.7} {ur}}|^2}{||\textbf{w}_{\scalebox{.7} {ur}}||^2\sigma^2}\nonumber\\
    &= \frac{\varphi_{\scalebox{.7} {ur}}P_{\scalebox{.7} {u}}||\textbf{h}_{\scalebox{.7} {ur}}||^2}{\sigma^2}.
\end{align}

\subsection{Harvested Energy}
As mentioned before, sER is carried out during the first phase of duration $\tau T$. In order to achieve the highest possible efficiency in sER, we adopt the scheme proposed in \cite{5}, which is shown in Fig.~\ref{fig:EH_block}. It works as follows: after receiving the signals, a phase shift is introduced to each signal by its corresponding phase shifting circuitry such that all the output signals are out of phase with each other. The outputs go through charge pumps composed of voltage boosting and rectifying circuits, after which they are added up to achieve the maximum combined direct-current (DC) signal\cite{5}. The recycled energy is used to assist the communication. The signal at the EH antennas is given by
\begin{figure}[t!]
    \centering
    \includegraphics[height=2.1in,width=3.2in]{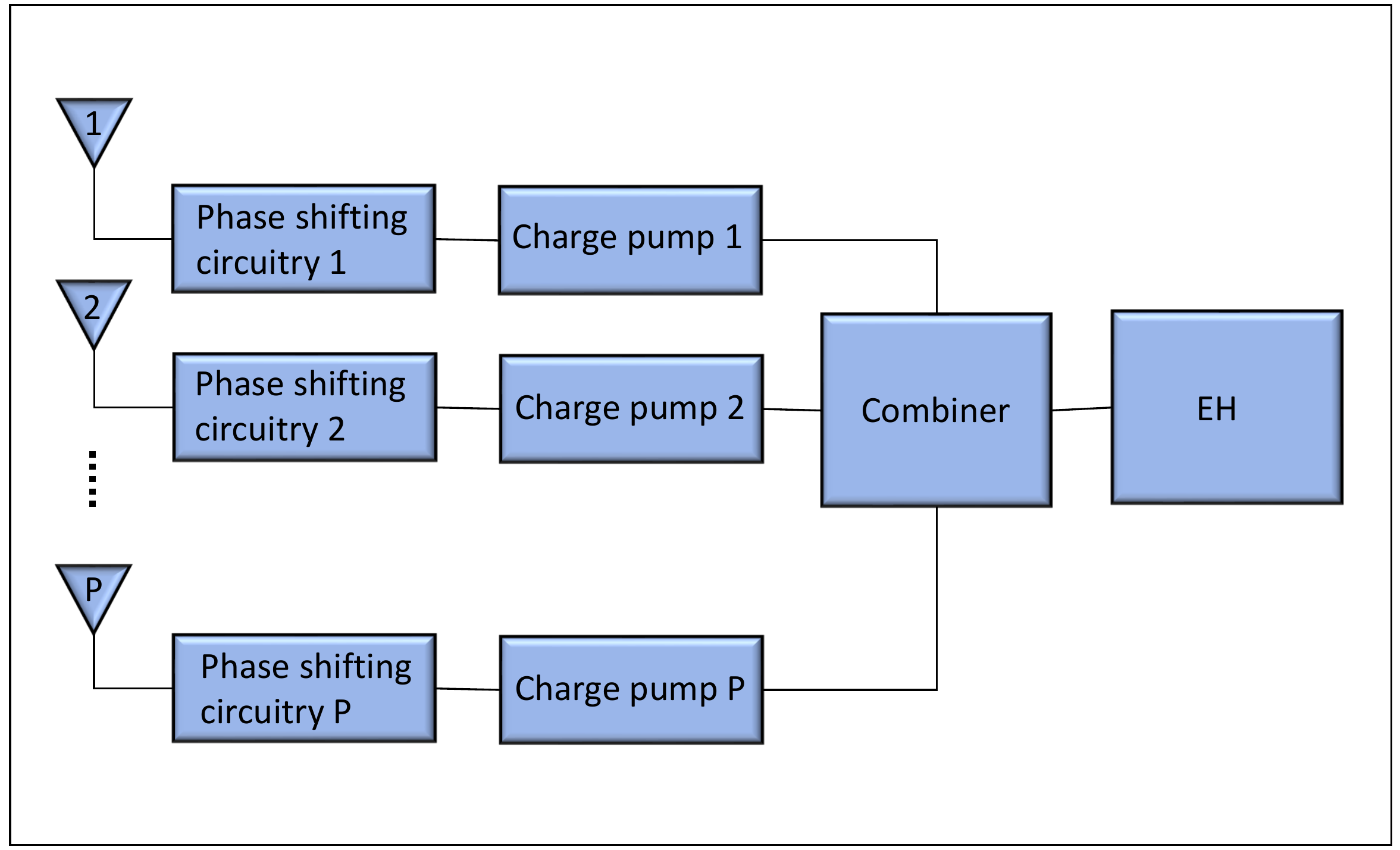}
    \caption{EH architecture. The design contains phase shifting circuits, charge pumps, a DC combiner and an EH block that allow to collect energy more efficiently.}
    \label{fig:EH_block}
\end{figure}
\begin{equation}
\textbf{y}_{\scalebox{.7} {p}}=\sqrt{\varphi_g(P_{\scalebox{.7} {EH}}+P_{\scalebox{.7} {RF}})}\textbf{w}_{\scalebox{.7} {td}}^H\textbf{G}s_s+\textbf{n}_{\scalebox{.7} {p}},
\end{equation}
where $\varphi_{g}$ represents the path gain in the link sER and$\textbf{\ n}_{\scalebox{.7} {p}}$ is the noise vector. $\textbf{G}$ is an $M\times P$ matrix containing all the channel coefficients in the link sER. These links will most likely experience near field rather than far field properties due to the short distance. Therefore, we may assume sER channel entries as equal and constant for performance analysis tractability \cite{lopez2020full} although this does not strictly hold in practice\footnote{Alternatively, these channels may be modeled with a Rician distribution with a very high line-of-sight (LOS) factor, which translates to almost zero randomness \cite{8} but over-complicating the analysis.}. 
 Clearly, the frequency used and the separation between elements in the antenna array impact these channels. The noise power is not considered for EH since its magnitude is too low compared to the self-interference level. Hence, the harvested energy is
\begin{align}
\label{eh}
E_{\scalebox{.7} {EH}}=&E_s\bigg\{\tau T\eta P\bigg|\sqrt{\varphi_{g}(P_{\scalebox{.7} {EH}}+P_{\scalebox{.7} {RF}})}\sum_{i=1}^{M} w_{\scalebox{.7} {td},i} g s_s\bigg|^2\bigg\}\nonumber\\
=&\tau T\eta P\varphi_{g}(P_{\scalebox{.7} {EH}}+P_{\scalebox{.7} {RF}})|g|^2\bigg|\sum_{i=1}^{M} w_{\scalebox{.7} {td},i}\bigg|^2,
\end{align}
where $\eta$ and $g$ are the conversion efficiency and channel coefficient in the link sER, respectively. We assume a linear model for the conversion efficiency for mathematical tractability. Readers may refer to \cite{lopez2020full,chen2017wireless} for  other nonlinear models that work better for low or high incident power levels. Moreover, $w_{\scalebox{.7} {td},i}$ represents the individual elements of the vector $\textbf{w}_{\scalebox{.7} {td}}$. Dividing both sides of (\ref{eh}) by $T$ we attain
\begin{align}
\label{eh_eq}
    P_{\scalebox{.7} {EH}} =&\tau\eta P\varphi_{g}(P_{\scalebox{.7} {EH}}+P_{\scalebox{.7} {RF}})|g|^2\bigg|\sum_{i=1}^{M} w_{\scalebox{.7} {td},i}\bigg|^2.
\end{align}
Notice that the term $P_{\scalebox{.7} {EH}}$ is at both sides of the equation, which means that there is an \emph{``energy loop''} since a portion of the transmit energy is captured and re-utilized as part of the sER process. It is worth mentioning that due to delay limitations, the harvested energy cannot be utilized immediately. However, after the transmission of a few symbols, the harvested energy stabilizes such that \eqref{eh_eq} holds. Now we can isolate $ P_{\scalebox{.7} {EH}}$ in (\ref{eh_eq}) as shown below
\begin{equation}
\label{Peh}
P_{\scalebox{.7} {EH}}=\cfrac{\eta P \varphi_{g}|g|^2P_{\scalebox{.7} {RF}}\big|\sum_{i=1}^{M} w_{\scalebox{.7} {td},i} \big|^2}{1-\tau\eta P \varphi_{g}|g|^2\big|\sum_{i=1}^{M} w_{\scalebox{.7} {td},i} \big|^2}.
\end{equation}
    Observe that $P_{\scalebox{.7} {EH}}$ is a random variable since it depends on $Z=|\sum_{i=1}^{M} w_{\scalebox{.7} {td},i}|^2=\frac{|\sum_{i=1}^{M} h_{\scalebox{.7} {td},i} |^2}{||\textbf{h}_{\scalebox{.7} {td}}||^2}$ which is a random variable itself. Here,  $h_{\scalebox{.7} {td},i}$ represents the individual elements of the vector $\textbf{h}_{\scalebox{.7} {td}}$. Notice that $Z$ is the ratio between two correlated exponential and gamma random variables whose distribution is very challenging to find analytically. Hence, we propose the approximation given in the following lemma.
   
    \begin{lemma}   
    \label{lemma_1}
    The random variable $Z$ approximately follows a generalized Pareto distribution (GPD) such that 
    \begin{align}
    \label{pdf}
    f_{Z}(z)\approx \frac{M-1}{M}\bigg(1-\frac{z}{M}\bigg)^{M-2}.
    \end{align}
    \end{lemma}
    \begin{proof}
    Refer to Appendix \ref{app_A}. \phantom\qedhere
    \end{proof}
    \begin{theorem}
    The distribution of $P_\mathrm{EH}$ is given by
    \begin{align}
    f_{P_\mathrm{EH}}(p_\mathrm{eh})\!=\!\frac{(M\!-\!1)a_1}{M(a_1\!+\!b_1p_\mathrm{eh})^2}\bigg(\!1\!-\!\frac{p_\mathrm{eh}}{M(a_1\!+b_1p_\mathrm{eh})}\bigg)^{\!\!M\!-\!2}\!\!\!\!\!\!,
    \end{align}
    where $a_1=\eta P \varphi_{g}|g|^2P_\mathrm{RF}$ and  $b_1=\tau\eta P \varphi_{g}|g|^2$.    \end{theorem}
    \begin{proof}
    According to (\ref{Peh}), the distribution of $P_\mathrm{EH}$ can be computed as a direct transformation of the variable $Z$ as
    \begin{align}
      P_{\scalebox{.7} {EH}}=\cfrac{a_1Z}{1-b_1Z}\nonumber
    \end{align}
    where $f_Z(z)$ is given in (\ref{pdf}). 
    \end{proof}
    \subsection{Outage Analysis}
    In this section, we define the outage events at the SBS and $D$ as $\mathbb{O}_{sbs}\overset{\mathrm{\triangle}}{=}\{(1-\tau)\log_2(1+\gamma_{\scalebox{.7} {sbs}})<r_{sbs}\}$ and $\mathbb{O}_{d}\overset{\mathrm{\triangle}}{=}\{\tau\log_2(1+\gamma_{\scalebox{.7} {d}})<r_{d}\}$, respectively. Therefore
    \begin{align}
    \mathbb{P}\{\mathbb{O}_{sbs}\}\overset{\mathrm{(a)}}{=}&\mathbb{P}\bigg\{||\textbf{h}_{\scalebox{.7} {ur}}||^2<\frac{(2^{\frac{r_{sbs}}{1-\tau}}-1)\sigma^2}{P_{\scalebox{.7} {u}}\varphi_{\scalebox{.7} {ur}}}\bigg\}\nonumber\\
    \overset{\mathrm{(b)}}{=}&\frac{1}{\Gamma(M)}\Gamma\bigg(M,\frac{(2^{\frac{r_{sbs}}{1-\tau}}-1)\sigma^2}{P_{\scalebox{.7} {u}}\varphi_{\scalebox{.7} {ur}}}\bigg),
    \end{align}
    where ($a$) comes from using \eqref{outage_sbs} and isolating $||\textbf{h}_{\scalebox{.7} {ur}}||^2$, and ($b$) comes from using the cumulative distribution function (CDF) of a gamma random variable. It is worth noting that the outage at the SBS in HD does not depend on the available power at the SBS. 
    
\begin{theorem}
\label{theor_2}
    The outage probability at $D$ when using HD at the SBS is given by
    \begin{align}
       \mathbb{P}&\{\mathbb{O}_{d}\}\!=\! 1\!-\!\frac{(M\!-\!1)!}{\mathrm{e}^{a_2}}\sum_{k=0}^{M-1}\sum_{j=0}^{k}\frac{a_2^{k-j}}{j!(k\!-\!j)!}(-b_2M)^j\Gamma(j\!+\!1)\nonumber\\
       &\qquad\qquad\qquad\qquad\ \ \ \times _1\!\!F_1(1+j,j+M,b_2M),\label{Od}
    \end{align}
    where $a_2=\frac{(2^{2r_d}-1)\sigma^2}{\tau P_\mathrm{RF}\varphi_\mathrm{td}}$
and $b_2 =\frac{(2^{2r_d}-1)\sigma^2P\eta\varphi_g}{P_\mathrm{RF}\varphi_\mathrm{td}}$.
\end{theorem}
\begin{proof}
\label{app_B}
We proceed as follows 
\begin{align}
&\mathbb{P}\{\mathbb{O}_{d}\}\nonumber\\
&=\mathbb{P}\bigg\{\tau\log_2(1+\gamma_{\scalebox{.7} {d}})<r_{d}\bigg\}
\nonumber\\
&\overset{\mathrm{(a)}}{=}\mathbb{P}\bigg\{2^{\frac{2r_{d}}{\tau}}-1>\frac{\varphi_{\scalebox{.7} {td}}(P_{\scalebox{.7} {EH}}+P_{\scalebox{.7} {RF}})||\textbf{h}_{\scalebox{.7} {td}}||^2}{\sigma^2}\bigg\}\nonumber\\
&\overset{\mathrm{(b)}}{=}\mathbb{P}\bigg\{\frac{(2^{\frac{2r_{d}}{\tau}}\!-\!1)\sigma^2(1\!-\!\tau\eta P \varphi_g |g|^2|\sum_{i=1}^M w_{\scalebox{.7} {td},i}|^2)}{P_{\scalebox{.7} {RF}}\varphi_{\scalebox{.7} {td}}}\!>||\textbf{h}_{\scalebox{.7} {td}}||^2\bigg\}\nonumber\\
&\overset{\mathrm{(c)}}{=}\mathbb{P}\{a_2-b_2z>||\textbf{h}_{\scalebox{.7} {td}}||^2\}\nonumber\\
&\overset{\mathrm{(d)}}{=}1-\!\!\int_{0}^{M}\!\!\!\!\frac{1}{\Gamma(M)}\Gamma(M,a_2\!-\!b_2z)\frac{M\!-\!1}{M}\bigg(1\!-\!\frac{z}{M}\bigg)^{M-2}\!\!\!\!\!\!\!\!dz\nonumber\\
&\overset{\mathrm{(e)}}{=} 1\!-\!\frac{(M\!-\!1)!}{\mathrm{M}}\!\!\sum_{k=0}^{M-1}\!\frac{1}{k!}\!\int_{0}^{M}\!\!\!\frac{(a_2\!-\!b_2z)^k}{\mathrm{e}^{a_2-b_2z}}\bigg(\!1\!-\!\frac{z}{M}\!\bigg)^{\!\!M-2}\!\!\!\!dz\nonumber\\
   &\overset{\mathrm{(f)}}{=}1-\frac{(M-1)!}{\mathrm{M}}\sum_{k=0}^{M-1}\sum_{j=0}^{k}\frac{1}{k!}\binom{k}{j}\frac{{a_2}^{(k-j)}}{\mathrm{e}^{a_2}}\nonumber\\
  &\qquad\qquad\qquad\times\int_{0}^{M}\mathrm{e}^{b_2z}(-b_2z)^j\bigg(1-\frac{z}{M}\bigg)^{M-2}\!\!\!dz,
\end{align}
where ($a$) comes from using (\ref{snr_1}),  ($b$) is obtained by using (\ref{Peh}) and isolating $||\textbf{h}_{\scalebox{.7} {td}}||^2$, ($c$) follows from substituting the values of $a_2$ and $b_2$ given in Theorem \ref{theor_2}, ($d$) is the integral of the CDF of a Gamma random variable weighted by (\ref{pdf}), ($e$) comes from applying $\Gamma(n+1,v)=n!\mathrm{e}^{-v}e_n(v)$ with $e_n(v)=\sum_{k=0}^{n}\frac{v^{k}}{k!}$ \cite[Eq. (8.352 2)]{jeffrey2007table}, ($f$)
is obtained by applying the binomial expansion to the term $(a_2-b_2z)^k$\cite[Eq. (1.111)]    {jeffrey2007table}, and finally \eqref{Od} is attained after computing the integral and performing some algebraic manipulations.\hfill 	
\end{proof}
 
 \subsection{Power Consumption}
 Notice that the outage at $D$ depends on the SBS transmit power. In practice this can be, at most, the power provided by the  source minus the power consumed by the elements in the RF chains and the amplifiers. Herein, we use the model proposed in \cite{7}, which states that the overall power consumption is mainly determined by two components: 1) the consumption of the power amplifiers $P_\mathrm{PA}$, and 2) the consumption of other circuit blocks $P_c$. The former can be computed as
\begin{equation}
    P_\mathrm{PA}=(1+\alpha)P_\mathrm{RF},
\end{equation}
where $\alpha=(\epsilon/\eta_{pa})-1$ such that $\eta_{pa}$ is the drain efficiency of the power amplifier and $\epsilon$ is the peak-to-average ratio (PAR). The second term is given by
\begin{align}
    P_c=M(P_{dac}+&P_{mix}+P_{filt})+2P_{syn}+N(P_{lna}+\nonumber\\&P_{mix}+
    P_{ifa}+P_{filr}+P_{adc}),
\end{align}
where $N$ is the number of receive antennas, while $P_{dac}, P_{mix}, P_{lna}, P_{ifa}, P_{filt}, P_{filr}, P_{adc}$ and  $P_{syn}$ are the power consumption values for the digital to analog converter (DAC), the mixer, the low-noise amplifier (LNA), the intermediate frequency amplifier (IFA), the active filters at the transmiter, the active filters at the receiver, the analog to digital converter (ADC) and the frequency synthesizer, respectively. Then, the transmit power can be computed as $P_\mathrm{RF}=(P_\mathrm{G}-P_c)/(1-\alpha)$.
The impact of considering this energy consumption as well as the EH process for different number of antennas is discussed in Section \ref{results}.
	\section{Full-Duplex}\label{section_FD}
	In the analysis of the FD scenario, we use a notation similar to that in the previous section unless stated otherwise. As shown in Fig. \ref{FD_figure}, the FD SBS\footnote{For additional information on FD radios, architectures and antenna configuration refer to \cite{alves2020full}.} is simultaneously transmitting and receiving information using $M$ and $N$ antennas such that $M+N=Q$. Here, the recycled energy is larger than in  the HD case since the recycling process occurs continuously. 
	\subsection{SNR in DL and UL Channels}
	Under this scheme, (\ref{Peh}) can still be used but with $\tau=1$. SIC is carried out in reception in order to mitigate the effect of the self-interference signal. 
	We consider at $D$ not only noise but also interference from $U$. Then, the signal at $D$ can be represented as 
	\begin{align}
	      	\text{y}_{\scalebox{.7} {d}}\!=\!\sqrt{\varphi_{\scalebox{.7} {td}}(P_{\scalebox{.7} {EH}}\!+\!P_{\scalebox{.7} {RF}})}\textbf{w}_{\scalebox{.7} {td}}^H\textbf{h}_{\scalebox{.7} {td}}s_s\!+\!\sqrt{\varphi_{\scalebox{.7} {ud}}P_{\scalebox{.7} {u}}}h_{\scalebox{.7} {ud}}s_\mathrm{u}\!+\!n_{\scalebox{.7} {d}},
    \end{align}
        where $\varphi_{\scalebox{.7} {ud}}$ depicts the path gain in the link UD, $h_{\scalebox{.7} {ud}}\sim \mathcal{CN}(0, 1)$ is the complex channel coefficient and $n_{\scalebox{.7} {d}}$ represents the AWGN with power $\sigma^2$, thus, the signal-to-interference-plus-noise ratio (SINR) is given by
    \begin{align}
            \gamma_{\scalebox{.7} {d}}&=\frac{|\sqrt{\varphi_{\scalebox{.7} {td}}(P_{\scalebox{.7} {EH}}+P_{\scalebox{.7} {RF}})}\textbf{w}_{\scalebox{.7} {td}}^H\textbf{h}_{\scalebox{.7} {td}}s_s|^2}{|\sqrt{\varphi_{\scalebox{.7} {ud}}P_{\scalebox{.7} {u}}}h_{\scalebox{.7} {ud}}s_u|^2+ \sigma^2}
            \nonumber\\
            &=\frac{\varphi_{\scalebox{.7} {td}}(P_{\scalebox{.7} {EH}}+P_{\scalebox{.7} {RF}})||\textbf{h}_{\scalebox{.7} {td}}||^2}{\varphi_{\scalebox{.7} {ud}}P_{\scalebox{.7} {u}}|h_{\scalebox{.7} {ud}}|^2+\sigma^2}. 
            \label{snr_2}
    \end{align}
At the SBS, the received signal is  
\begin{figure}[t!]
    \centering
    \includegraphics[height=2.1in,width=3.2in]{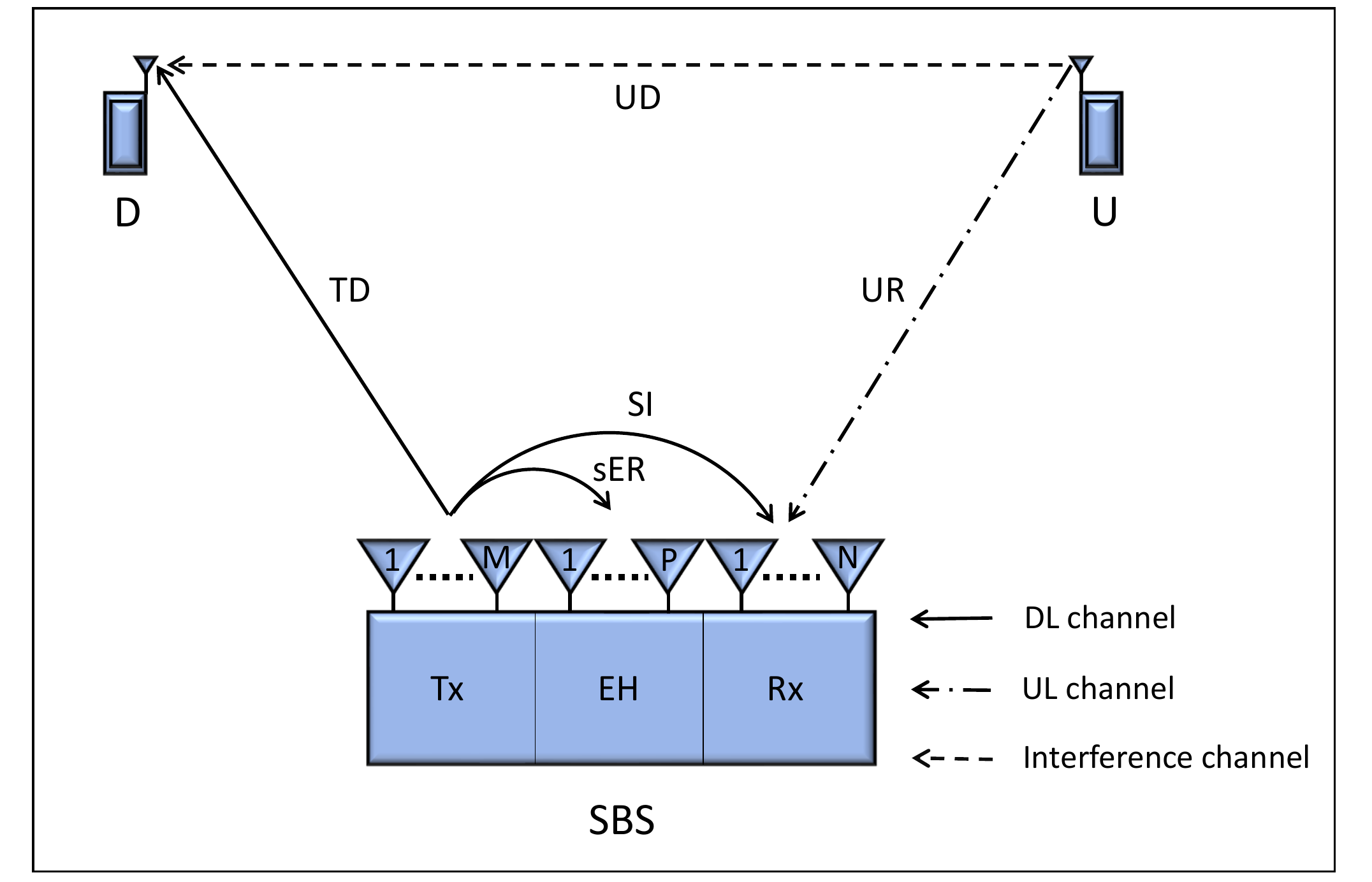}
    \caption[Optional caption]{Communication system operating in FD mode. The SBS communicates with $D$ and $U$ through the DL and UL channels, respectively. It is also equipped with $Q$ RF chains and P antennas for EH. The figure shows the self-interference channel (SI) between the transmit and receive antennas.}
    \label{FD_figure}
    \end{figure}
\begin{equation}
\textbf{y}_{\scalebox{.7} {sbs}}\!\!=\!\sqrt{\varphi_{\scalebox{.7} {ur}}P_{\scalebox{.7} {u}}}\textbf{h}_{\scalebox{.7} {ur}}s_\mathrm{u}\!+\!\sqrt{\varphi_{\scalebox{.7} {SI}}(P_{\scalebox{.7} {EH}}\!+\!P_{\scalebox{.7} {RF}})\zeta}\textbf{w}_{\scalebox{.7} {td}}^H\textbf{G}_s s_s\!+\!\textbf{n}_{\scalebox{.7} {r}},
\end{equation}
where $\textbf{G}_s$ is a matrix containing the near field channel coefficients between the transmit and receive antennas at the SBS, and $\varphi_{\scalebox{.7} {SI}}$ is the path gain in the self-interference channel. Meanwhile, $\textbf{h}_{\scalebox{.7} {ur}}$ and $\textbf{n}_{\scalebox{.7} {r}}$ denote the UL channel coefficient vector and the AWGN at the SBS's receive antennas, respectively, while the SIC coefficient is denoted by $\zeta$. Reducing the self-interference requires either passive (antenna separation) or active (analog, digital) techniques. Different models have been proposed in the literature for capturing the effect of the residual self-interference, for instance, the authors in \cite{8}, modeled the residual interference as a Rician-distributed random variable with large (passive SIC) and small (active SIC) LOS factor. Herein, we assume a good enough and constant SIC coefficient ($\leqslant-90$dB) to take the interference to the noise level\footnote{SIC can easily achieve such levels as observed in \cite{zhang2016full} and \cite{heino2015recent}.}.
 Then, the SINR is expressed as 
\begin{align}
\gamma_{\scalebox{.7} {sbs}}=&\frac{|\sqrt{\varphi_{\scalebox{.7} {ur}}P_{\scalebox{.7} {u}}}\textbf{w}_{\scalebox{.7} {ur}}^H\textbf{h}_{\scalebox{.7} {ur}}s_u|^2}{|\sqrt{\varphi_{\scalebox{.7} {SI}}(P_{\scalebox{.7} {EH}}\!+\!P_{\scalebox{.7} {RF}})\zeta}\textbf{w}_{\scalebox{.7} {ur}}^H\textbf{G}_s^H\textbf{w}_{\scalebox{.7} {td}}s_s|^2\!+\!E\{|\textbf{w}_{\scalebox{.7} {ur}}^H\textbf{n}_{\scalebox{.7} {u}}|\}^2}\nonumber\\
=&\frac{\varphi_{\scalebox{.7} {ur}}P_{\scalebox{.7} {u}}||\textbf{h}_{\scalebox{.7} {ur}}||^2}{\varphi_{\scalebox{.7} {SI}}(P_{\scalebox{.7} {EH}}+P_{\scalebox{.7} {RF}})\zeta|\textbf{w}_{\scalebox{.7} {ur}}^H\textbf{G}_s ^H\textbf{w}_{\scalebox{.7} {td}}|^2+\sigma^2}.
\label{snr_3}
\end{align} 
\subsection{Outage analysis}
\begin{theorem}
\label{theorem_3}
The outage probability at $D$ when using FD at the SBS is given by
\begin{align}
\mathbb{P}\{\mathbb{O}_{d}\}=1-\mathrm{e}^{a_3}(1-(1+b_3)^{-M}),
\end{align}
where $a_3=\frac{\sigma^2}{\varphi_\mathrm{ud}P_\mathrm{u}}$ and $b_3=\frac{\varphi_\mathrm{td}P_\mathrm{RF}}{\varphi_\mathrm{ud}P_\mathrm{u}(2^{r_{d}}-1)(1-\eta P \varphi_{g}|g|^2)}$.
\end{theorem}
\begin{proof}
We proceed  as follows in order to find the outage probability at $D$
\begin{align}
\mathbb{P}\{\mathbb{O}_{d}\}&=\mathbb{P}\{\log_2(1+\gamma_{\scalebox{.7} {d}})<r_{d}\}
\nonumber\\
&\overset{\mathrm{(a)}}{=}\mathbb{P}\bigg\{2^{r_{d}}\!-\!1\!>\nonumber\\
&\qquad\ \frac{\varphi_{\scalebox{.7} {td}}P_{\scalebox{.7} {RF}}||\textbf{h}_{\scalebox{.7} {td}}||^2}{(\varphi_{\scalebox{.7} {ud}}P_{\scalebox{.7} {u}}|h_{\scalebox{.7} {ud}}|^2\!\!+\!\!\sigma^2)(1\!\!-\!\!\tau\eta P \varphi_{g}|g|^2\big|\sum_{i=1}^{M} w_{\scalebox{.7} {td},i}
\big|^2)}\bigg\}\nonumber\\
&\overset{\mathrm{(b)}}{=}\mathbb{P}\bigg\{\frac{|h_{\scalebox{.7} {ud}}|^2\!\!+\!\!\frac{\sigma^2}{\varphi_{\scalebox{.7} {ud}}P_{\scalebox{.7} {u}}}}{||\textbf{h}_{\scalebox{.7} {td}}||^2} >\nonumber\\
&\qquad\ \frac{\varphi_{\scalebox{.7} {td}}P_{\scalebox{.7} {RF}}}{\varphi_{\scalebox{.7} {ud}}P_{\scalebox{.7} {u}}(2^{r_{d}}\!\!-\!\!1)(1\!\!-\!\!\tau\eta P \varphi_{g}|g|^2\big|\!\sum_{i=1}^{M} w_{\scalebox{.7} {td},i} \big|^2)}\bigg\}\nonumber\\
&\overset{\mathrm{(c)}}{\approx}1-\frac{1}{\Gamma(M)}\int_{0}^{b_3}\!\!\int_{0}^{\infty}\!\!\mathrm{e}^{(a_3-Ry)}\mathrm{e}^{-y}y^{M}dydR\nonumber\\
&\overset{\mathrm{(d)}}{=}1-\mathrm{e}^{a_3}(1-(1+b_3)^{-M})\nonumber,
\end{align}
where  ($a$) is obtained by using (\ref{Peh}) and (\ref{snr_2}), ($b$) follows after some algebraic manipulations in order to have the ratio $(R=X/Y)$ of a shifted exponential and a gamma random variable, while ($c$) is an approximation since we ignored the term $\big|\sum_{i=1}^{M} w_{\scalebox{.7} {td},i} \big|^2$ in the denominator of ($b$) because  $\tau\eta P \varphi_{g}|g|^2$ is usually small. Then, we applied the concept of ratio distribution, e.g., $p_R(r)=\int_{0}^{\infty} y p_{X,Y}(ry,y) dy$. Finally, ($d$) is attained after computing the double integral. \hfill 
\end{proof}
\begin{theorem}
\label{theorem_4}
The outage probability at the SBS when using FD is given by
\begin{align}
   \mathbb{P}\{\!\mathbb{O}_{sbs\!}\}\!&=\! 1\!-\!\beta(M\!\!-\!\!1,\!N\!\!-\!\!1)\!\!\!\sum_{i=0}^{N-1}\!\sum_{l=0}^{i}\!\!\sum_{p=0}^{M\!+\!N\!-\!3}\!\frac{b_4^la_4^{i-l}(M\!\!-\!1)(N\!\!-\!1)}{i!\mathrm{e}^{b_4}}\nonumber\\
	&\qquad\qquad\times(-1)^p{M\!+\!N\!-\!3\choose p}{i\choose l}(MN)^{i-l}\nonumber\\
	&\qquad\qquad\times\bigg(\sum_{s=1}^{n}w_s\big[\mathrm{e}^{u}f(u)\big]\!+\!R_n(u)\bigg),
\end{align}
where 
\begin{align}
a_4=\frac{(2^{r_{d}}\!-\!1)\zeta|g_s|^2\big|P_\mathrm{RF}\varphi_\mathrm{SI}}{P_\mathrm{u} \varphi_\mathrm{ur}(1-\eta P \varphi_{g}|g|^2)},\:\:\: b_4=\frac{\sigma^2}{P_\mathrm{u}\varphi_\mathrm{ur}},\nonumber
\end{align}
$u=-1+\frac{MN}{q}$ and $q = \big|\sum_{i=1}^{N}\! w_{\mathrm{ur},i} \big|^2\big|\sum_{i=1}^{M}w_{\mathrm{td},j} \big|^2$.
Moreover, $g_s$ represents the channel coefficients of matrix $\mathbf{G}_s$, and the terms $w_s$ and $R_n(u)$ are defined as
\begin{align}
\label{Wk}
    w_s=& \frac{(n!)u_s}{(n+1)^2\big[L_{n+1}(u_s)\big]^2},\\ R_n(u)&=\frac{(n!)^2}{(2n)!}f^{(2n)}(u),
    \label{Rn}
	\end{align}
	   where $L_{n}(u_k)$ are the Laguerre polynomials, $u_s$ is the $s-$th root of the polynomials, and $f(u)$ is given by
	   \begin{align}
	      f(u)=&\mathrm{e}^{-\frac{a_4MN}{1+u}}(1+u)^{-3-i+l+N-p}\nonumber\\
	      &\times{_2F_1}(N-1,N-1,M+N-2,-u). \label{fu}
	   \end{align}
\end{theorem}
\begin{proof}
We proceed as follows 
%
\begin{align}
&\mathbb{P}\{\mathbb{O}_{sbs}\}\!=\!\mathbb{P}\{\log_2(1\!+\!\gamma_{\scalebox{.7} {sbs}})\!<\!r_{sbs}\} =\mathbb{P}\{||\textbf{h}_{\scalebox{.7} {ur}}||^2<a_4q+b_4\}\nonumber\\
&\overset{\mathrm{(a)}}{=}1\!-\frac{1}{\Gamma(N)}\int_{0}^{MN}\!\!\!
\Gamma(N,a_4q\!+\!b_4)(M\!-\!1)(N\!-\!1)\bigg(
\frac{1}{M}\bigg)^{\!M\!-\!1}\!\!\!\nonumber\\
&\qquad\times\bigg(\frac{q}{N}\bigg)^{\!\!\!-N+2}\!\bigg(M\!-\!\frac{q}{N}\bigg)^{\!\!M+N-3}\!\!\!\!\!\beta(M\!-\!1,N\!-\!1)\nonumber\\
&\qquad\times _2\!F_1\bigg(N\!-\!1,N\!-\!1,M\!+\!N\!-\!2,1\!-\!\frac{MN}{q}\bigg)dq\nonumber\\
	&\overset{\mathrm{(b)}}{=}1\!-\!\beta(M\!-\!1,N\!-\!1)\!\sum_{i=0}^{N\!-\!1}\sum_{l=0}^{i}\!\sum_{p=0}^{M\!+\!N-3}\!\frac{b_4^la_4^{i-l}(M\!-\!1)(N\!-\!1)}{i!\mathrm{e}^{b_4}}\nonumber\\
	&\qquad\times\!(-1)^p{M\!+\!N\!-\!3\choose p}{i\choose l}\!\bigg(\frac{1}{MN}\bigg)^{-N+p+2}\!\!\!\int_{0}^{MN}\!\!\!\!\!\!\!\!\mathrm{e}^{-a_4q}\nonumber\\
	&\qquad\times q^{i\!-\!l\!-\!N\!+\!p\!+\!1}  {_2F_1}\bigg(\!N\!-\!1,N\!-\!1,M\!+\!N\!-\!2,1\!-\!\frac{MN}{q}\!\bigg)dq
	\nonumber\\
	&\overset{\mathrm{(c)}}{=}1\!-\!\beta(M\!-\!1,N\!-\!1)\sum_{i=0}^{N-1}\sum_{l=0}^{i}\!\sum_{p=0}^{M\!+\!N\!-\!3}\!\frac{b_4^la_4^{i-l}(M\!-\!1)(N\!-\!1)}{i!\mathrm{e}^{b_4}}\nonumber\\
	&\qquad\times(-1)^p{M\!+\!N\!-\!3\choose p}{i\choose l}(MN)^{i-l}\!\!\int_{0}^{\infty}\!\mathrm{e}^{-\frac{a_4MN}{1\!+\!u}}\nonumber\\
	&\qquad\times\!(1\!+\!u)^{-3-i+l+N-p}{_2F_1}(N\!\!-\!\!1,N\!\!-\!\!1,M\!\!+\!\!N\!\!-\!2,-u)du, \label{Osbs}
	\end{align}
 where ($a$) follows from stating the CDF of the Gamma distribution evaluated on $a_4q+b_4$, and averaged out over the PDF of $q$, which is given in \cite[Theorem 1]{9}. ($b$) comes from using \cite[Eq. (1.111)]    {jeffrey2007table}, and also from applying the binomial expansion to the term $\big(M\!-\!\frac{q}{N}\big)^{M\!+\!N\!-\!3}$ since the exponent is integer.
 	The remaining integral in ($b$) is difficult to solve analytically due to the hypergeometric function.
 	We address this issue by changing the integration limits to $[0, \infty]$ through a change of variable, setting $u$ as the new one as in ($c$). Thus, the integral can be put in the form $\int_{0}^{\infty}e^{-u}\big[e^{u}f(u)\big]du$  and computed using the Gauss-Laguerre (GL) method \cite[Eq.(25.4.45)]{abramowitz1948handbook} to attain \eqref{Osbs}. Note that $f(x)$ is the integrand of the integral after changing the limits, thus it is given by \eqref{fu}. The definition of the GL formula is as follows
	\begin{align}
	    \int_{0}^{\infty}\!\!\!\!\!\mathrm{e}^{-x}\big[\mathrm{e}^{x}f(x) dx\big]=\sum_{s=1}^{n}w_s\big[\mathrm{e}^{x}f(x)\big]+R_n(x),
	    \label{laguerre}
\end{align}
where $ w_s$ and $R_n(x)$ are defined in \eqref{Wk} and \eqref{Rn}, respectively. Note that (\ref{laguerre}) is exact if the residual can be calculated, which is strongly dependent on the existence of the high-order derivatives of $f(x)$. The accuracy and convergence of the GL method for computing the integral is shown in Appendix~\ref{convergence}. 
\end{proof}
\vspace{-3mm}
\subsection{Dynamic FD Scheme}
	   Now, since expressions for the outage probability are available for each scenario, we 
	   can
	   dynamically adjust the  set of transmit and receive antennas according to the performance metric of interest. In this case, we aim at minimizing the maximum outage probability among all links for maximum fairness, i.e., MinMax outage performance.  Note that the information regarding the adopted antenna configuration is not required at $U$/$D$ since they are single-antenna devices and cannot use any precoding/combining mechanism. The optimization problem for implementing the optimal FD scheme is given by
	   \begin{subequations}\label{P1}
	\begin{alignat}{2}
	\mathbf{P1:}\qquad &\underset{M,N}{\mathrm{min.}}       &\ \ \ & 
	\max_{i\in\{D,SBS\}}\
	\mathbb{P}\{\mathbb{O}_{i}\} \label{P1:a}\\
	&\text{s.t.}   &      & M+N=Q\label{P1:b}\\
		& & & M,N \geq 2\label{P1:c}
	\end{alignat}
\end{subequations}
In addition, we can set a certain reliability target and determine the minimum number of RF chains required to achieve it as 
\begin{subequations}\label{P2}
	\begin{alignat}{2}
	\mathbf{P2:}\qquad &\underset{M,N}{\mathrm{min.}}       &\ \ \ & 
	Q \label{P2:a}\\
	&\qquad\ \ \ \ \text{s.t.} &      & M+N=Q\label{P2:b}\\
		& & & M,N \geq 2\label{P2:c}\\
		& & & \mathbb{P}\{\mathbb{O}_{i}\}\leq\delta \label{P2:D}
	\end{alignat}
\end{subequations}
where $\delta$ is the targeted value.
We can proceed to find the solution to $\mathbf{P1}$ and $\mathbf{P2}$ using an ``Exhaustive search". This is a valid approach even if we are in presence of a large antenna array since the number of combinations to be tested is $Q-2$. Different from other state-of-the-art formulations, e.g., \cite{agrawal2020noma,guo2019performance}, these problems aim at dynamically distributing the antenna set between transmission and reception instead of adding more antennas to both functions or increasing the transmit power, which might be difficult to implement due to hardware  limitations. Notice that the proposed problems require changes on the antenna configuration, which would cause variations on the SIC performance, mainly on passive techniques based on antenna separation. Nevertheless, small variations would not affect the performance if the overall SIC scheme is strong enough to achieve more than 90-100 dB of cancellation. Numerical solutions to $\mathbf{P1}$ and $\mathbf{P2}$ are presented in the next section.
\vspace{-2mm}
\section{Numerical Results}\label{results}
Table \ref{table_2} shows the system parameter values utilized in this section unless stated otherwise. The power consumption values of the active elements of the RF chains are taken from \cite{7} and the used transmission rates are typical in multi-antenna techniques where higher values can be reached in the high SNR regime. We used a symmetric HD scheme for the simulations, hence $\tau=0.5$, while the rest of the parameters used for the simulations are practical for real systems. The value used for the SIC factor relies on the assumption that an efficient technique is used. Moreover, we assume that $U$ and $D$ are far enough in such a way that we can ignore the effect of the interference at $D$, allowing a fair comparison between HD and FD scenarios.
\vspace{-2mm}
\subsection{Accuracy of the Derived Outage Probability Expressions}
We start by checking the accuracy of the attained analytical expressions for the outage probability in both HD and FD scenarios. We compare two cases: i) ideal case, where the power consumed by active elements in the RF chains is not considered, i.e., $P_\mathrm{RF}=P_\mathrm{G}$, and ii) practical case, where such power consumption is taken into account, i.e., $P_\mathrm{RF}=(P_\mathrm{G}-P_c)/(1-\alpha)$. 

\begin{table}[t!]
    \centering
    \caption{Simulation parameters.}
    \label{table_2}
    \begin{tabular}{l  c  c}
        \thickhline
            \textbf{Parameter} & \textbf{Value} & \\
        \hline
            $Q$ & $\{8,16\}$ & \\
            $P_{dac}, P_{adc}$ & 1 mW & \cite{7} \\
            $P_{mix}$ & 30.3 mW & \cite{7} \\
            $P_{lna}$ & 20 mW & \cite{7} \\
            $P_{ifa}$ & 3 mW & \cite{7} \\
            $P_{filt}, P_{filr}$ & 2.5 mW & \cite{7} \\
            $P_{syn}$ & 50 mW & \cite{7} \\
            $\sigma^2$ & 0.1 nW &  \\
            $\eta$ & 0.6 & \\
            $P_\mathrm{G}$ & 15 W & \\
            $P$ & \{0, 6\} & \\
            $P_{\scalebox{.7} {u}}$ & 200 mW & \\
            $\zeta$ & -100 dB &\cite{zhang2016full,heino2015recent}\\
            $\varphi_g$ & -15 dB & \cite{lopez2020full}\\
            $\varphi_\mathrm{ud}$ & -60 dB &\\
        \thickhline
    \end{tabular}
\end{table}
\begin{figure}[t!]
    \centering
    \includegraphics[height=2.35in,width=3.3in]{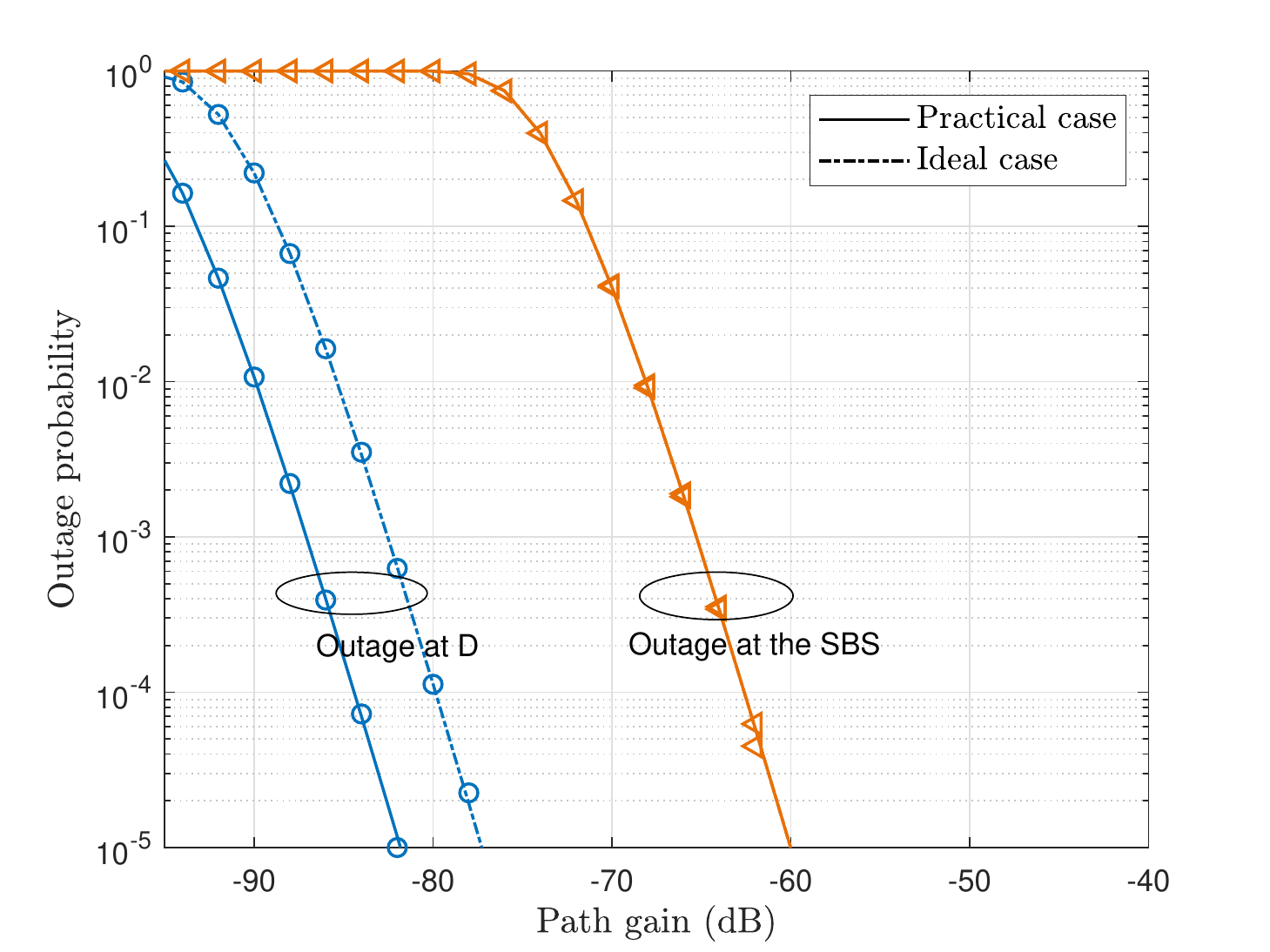}
    \centering
    \includegraphics[height=2.35in,width=3.3in]{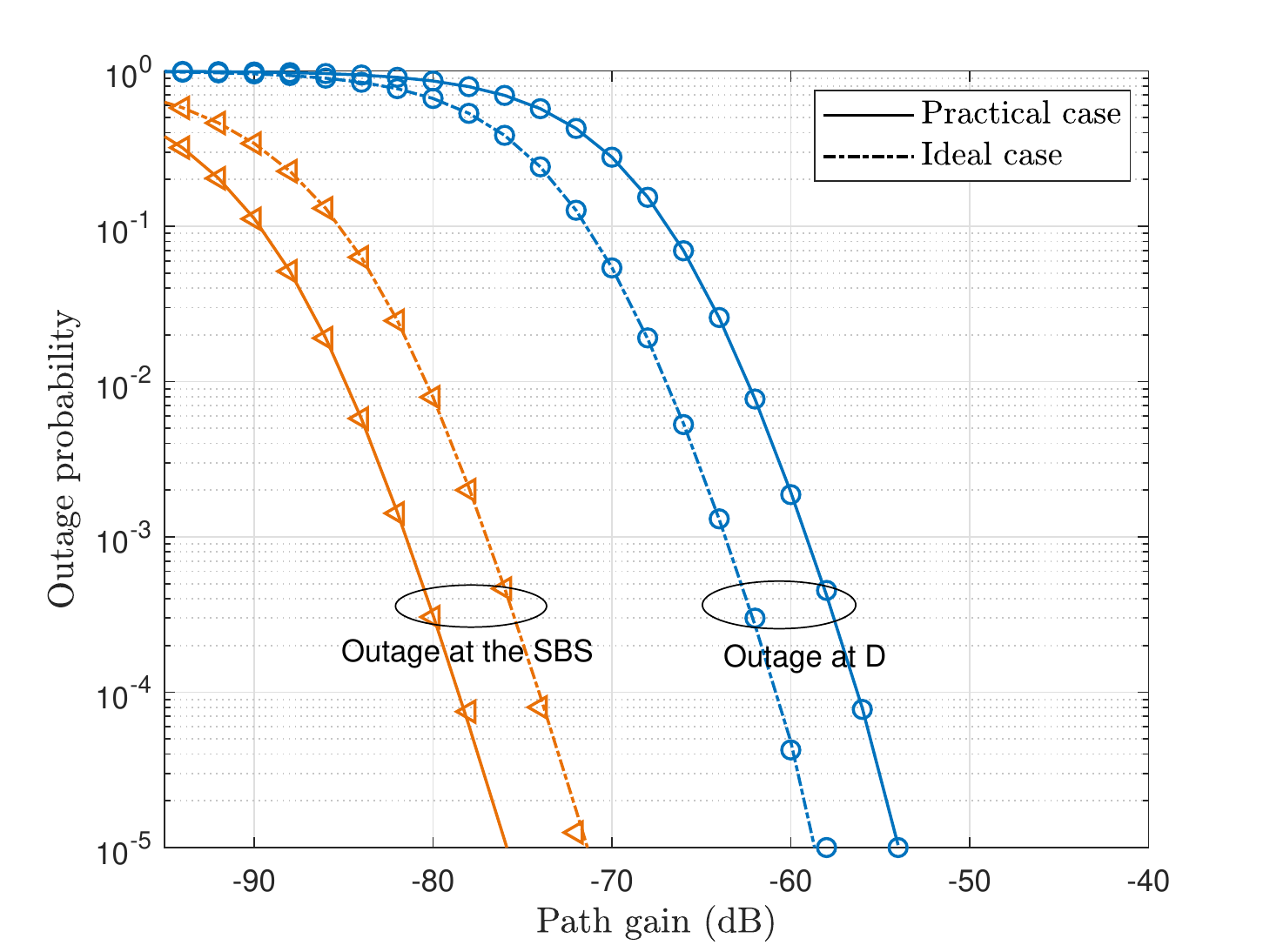}
    \caption{Outage probability as function of $\varphi_\mathrm{td}$ and $\varphi_\mathrm{ur}$. (a) HD (top). (b) FD (bottom). $\varphi_\mathrm{ud}=$ -60 dB. Markers represent Monte Carlo simulations in both figures.}
    \label{fd_outage}
    \vspace{-6mm}
\end{figure}

\begin{figure}[t!]
    \centering
    \includegraphics[height=2.35in,width=3.3in]{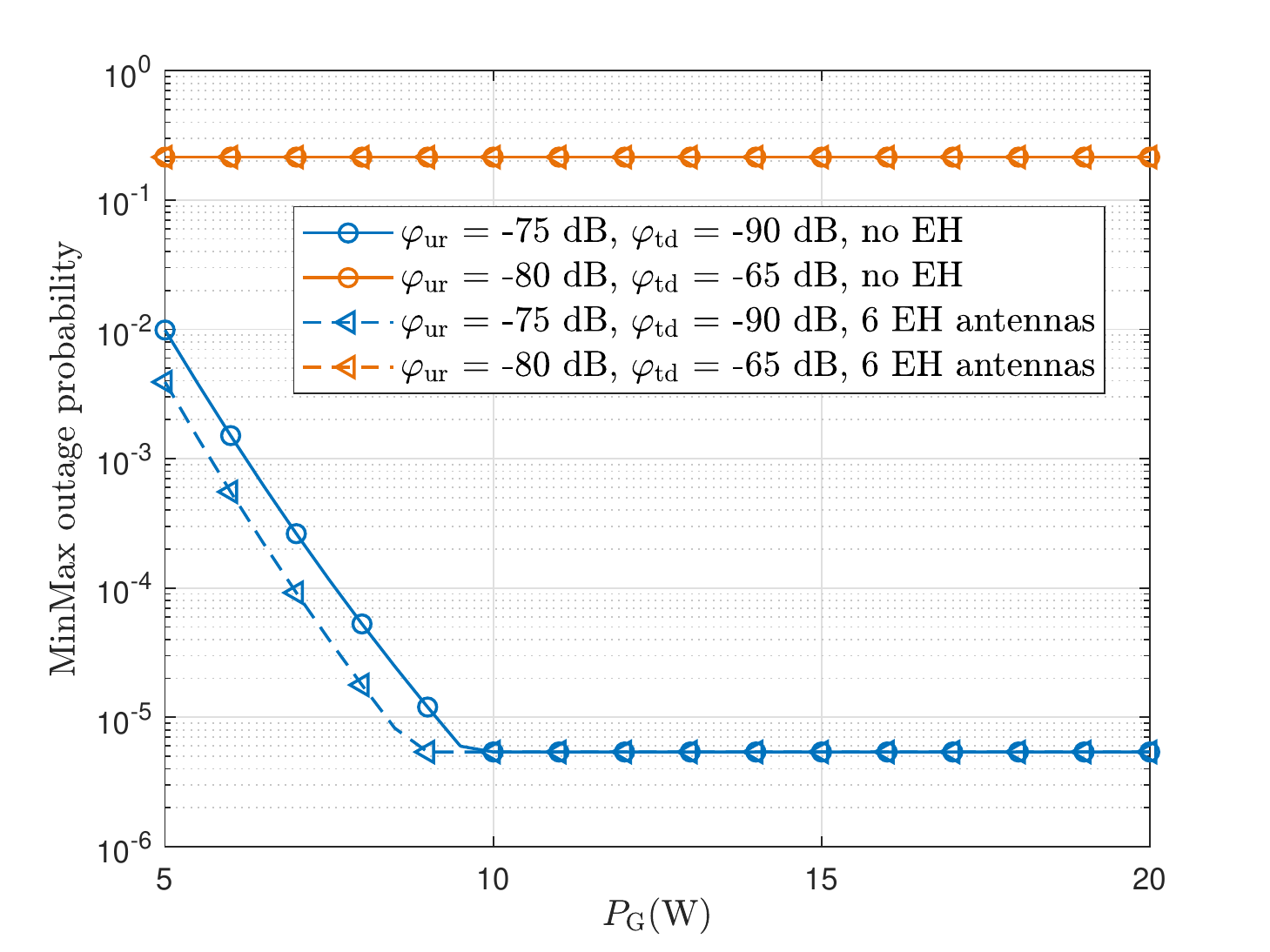}
    \centering
    \includegraphics[height=2.35in,width=3.3in]{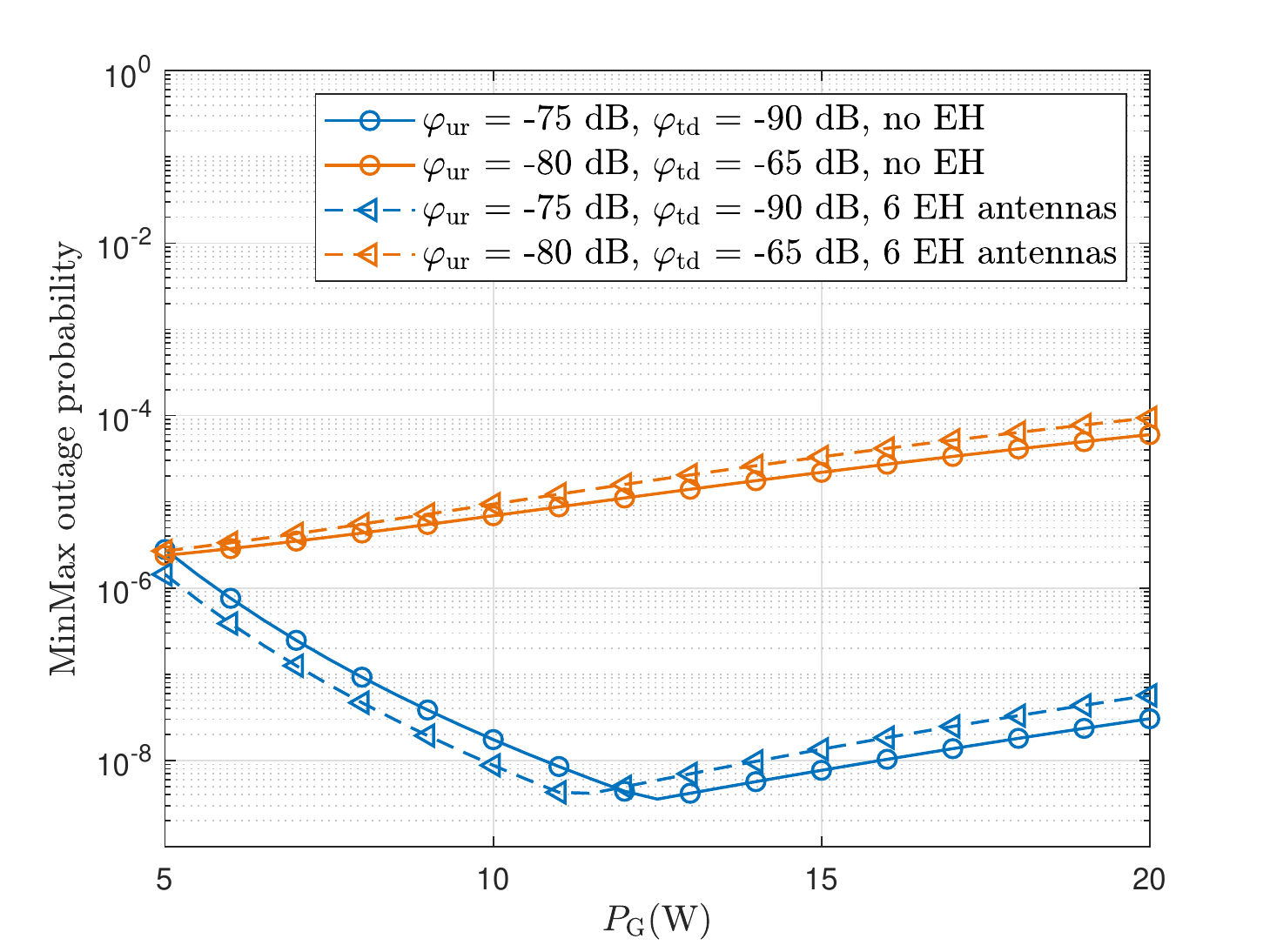}
    \caption{MinMax outage probability with fixed configuration in HD (top, $M = 16$) and FD (bottom, $M = N = 8$) with $r_d = r_{sbs}$ = 4 bps/Hz.}
    \label{fixed_conf}
    \vspace{-6mm}
\end{figure}
  Fig. \ref{fd_outage} shows the outage performance for HD and FD scenarios. In Fig. \ref{fd_outage} (a), it can be noticed a considerable performance gap in the link TD between the practical and the ideal HD case. Notice that when the circuit power consumption in the RF chains is considered, there is less power available for transmission; however, this consumption does not harm the outage probability at the SBS at all.
  Fig. \ref{fd_outage} (b) shows the same comparison at $D$ but for the FD scenario. Similarly, the performance is worse when considering the practical case since the signal is transmitted with less power. On the other hand, and differently from the HD scenario, the performance at the SBS is better when comparing to the ideal case. This is because the lower the transmission power, the lower the self-interference level, which translates to an improvement in terms of outage.\newline
\begin{figure}[t!]
    \centering
    \includegraphics[height=2.35in,width=3.3in]{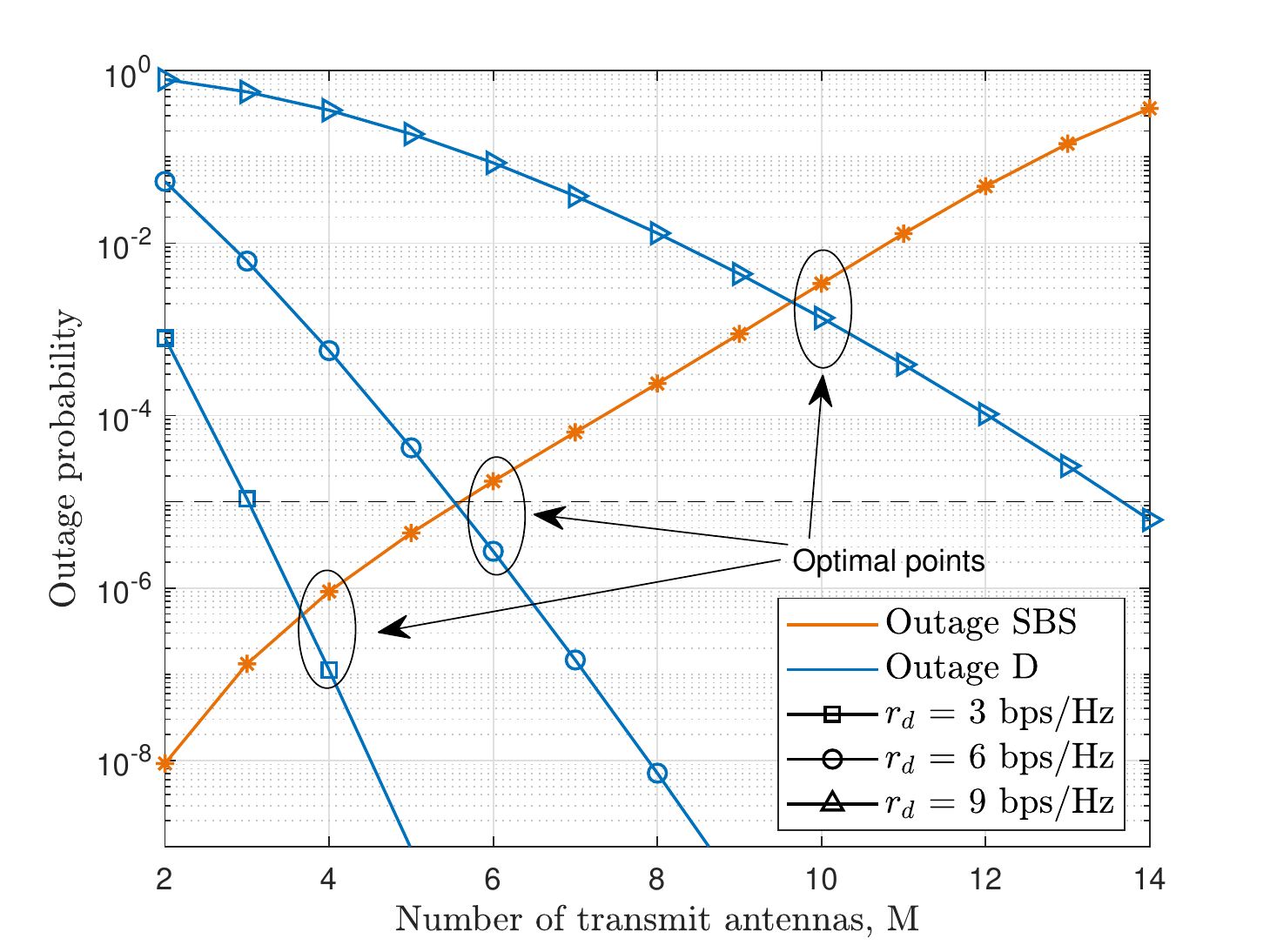}
    \centering
    \includegraphics[height=2.35in,width=3.3in]{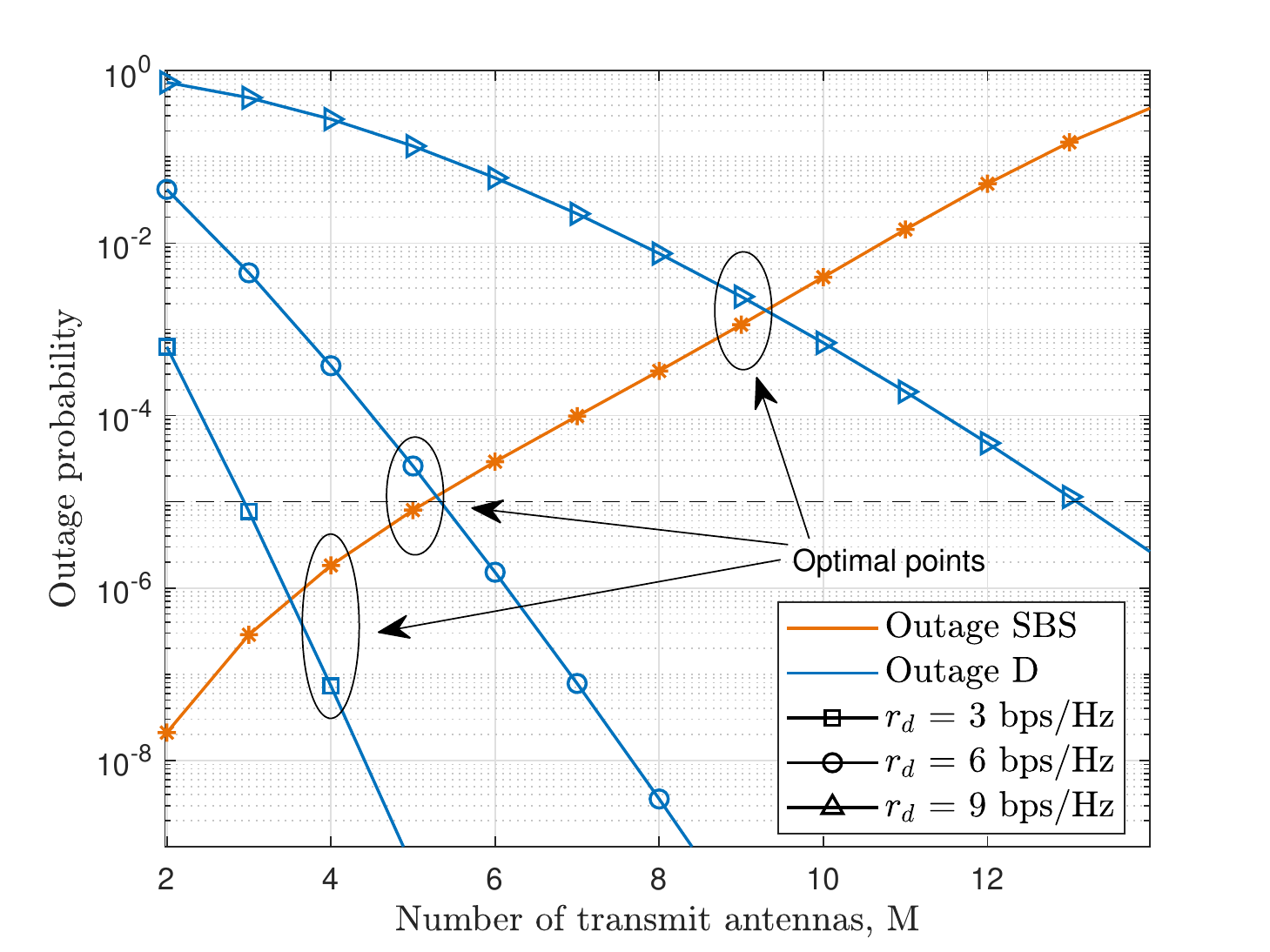}
    \caption{Optimal points FD. (a) No EH (top). (b) 6 EH antennas) (bottom). $r_{sbs}$ = 3 bps/Hz and $\varphi_\mathrm{ur}=\varphi_\mathrm{td}= -80$ dB. The area below the dashed lines represents the URC region.}
    \label{fig:my_label}
    \vspace{-6mm}
\end{figure}
\vspace{-2mm}
\subsection{Performance of Fixed FD and HD Schemes}
\vspace{-2mm}
\indent Fig. \ref{fixed_conf} shows how HD and FD perform in terms of MinMax outage probability for fixed configurations and equal rate in UL and DL when varying the power provided by the source $P_\mathrm{G}$. Here, $U$ and $D$ would probably be at different distances from the SBS  since path gains are considered different between them. In HD,  increasing the power provided by the source for $\varphi_\mathrm{ur} = -75$ dB and $\varphi_\mathrm{td} = -90$ dB without sER is beneficial until 9.5 W, approximately since from this point UR becomes the worst link. For the same scenario but using 6 EH antennas, we can see that a smaller value of $P_\mathrm{G}$ is required to achieve the same performance. This is because some extra power comes from the sER process. For $\varphi_\mathrm{ur} = -80$ dB and $\varphi_\mathrm{td} = -65$ dB,  plots show that in this scenario the worst link is UR all the time, since this one does not depend on $P_\mathrm{G}$. In FD, on the other hand, we can see that all schemes are influenced by $P_\mathrm{G}$. Similar to HD, for $\varphi_\mathrm{ur} = -75$ dB and $\varphi_\mathrm{td} = -90$ dB with 0 EH antennas / 6 EH antennas, we can see how up to approximately 12.5/11.5W TD is the worst link, but from these points on, UR becomes the worst, since self-interference grows with $P_\mathrm{G}$. Hence, sER is beneficial here in terms of MinMax outage probability under $P_\mathrm{G} <$ 12.5 W, for which the plot with 0 EH antennas reaches its minimum value.
\vspace{-2mm}
\subsection{Performance of the Dynamic FD Scheme}
\indent Fig. \ref{fig:my_label} shows the maximum outage probability between UL and DL as a function of the number of transmit antennas and for different values of $r_d$. Obviously, there is an optimum operational point which guarantees the best minimum level of reliability among the devices in the network. It can be noticed how the optimal combination of transmit (x-axis) and receive antennas (total$-$transmit) is the one with closest points between both curves. This is because since we are distributing the antennas between transmission and reception, we cannot decrease the outage events at both devices simultaneously. Then, a larger number of antennas for transmission will decrease the outage probability at $D$ but will increase it at the SBS. It is worth noting the effect of sER in the performance. We can see in Fig. \ref{fig:my_label} (a) how the optimal number of transmit antennas are 4, 6 and 10  for the 3 different $r_{d}$ values without sER, whereas 3, 5 and 9 are the optimal ones when 6 antennas are used for this purpose as shown in Fig. \ref{fig:my_label} (b). This is because with recycling antennas, the available power is $P_\mathrm{RF}+P_\mathrm{EH}$, which causes a variation in the probabilities and therefore different optimal combinations. The number of transmit and receive antennas will also depend on the service requirements. For instance, it can be seen that if the requirements are close to $10^{-6}$ with $r_d$ = 6 bps/Hz, the combination that achieves the best performance is 5 transmit and 11 receive antennas. On the other hand, if the required outage probability changes to $10^{-3}$ under the same network configuration, then, in addition to the optimal combination, 4 to 9 transmit antennas can be used since both $D$ and the SBS will fulfill the requirements. Finding the optimal number of Tx and Rx antennas brings fairness between both links, this is because they can operate with a performance as close as possible to each other. In Fig. \ref{fig:my_label} (a), for $r_d$ = 6 bps/Hz, we can notice that if we fix $M = N =$ 8, the outage probabilities are in the order of $10^{-2}$ at the $SBS$ and $10^{-8}$ at D, however, using the optimal configuration of $M =$ 6 and $N =$ 10 both links operate near the URC region simultaneously, i.e., achieving outage probabilities around $10^{-5}$.
\begin{figure}[t!]
    \centering
    \includegraphics[height=2.35in,width=3.3in]{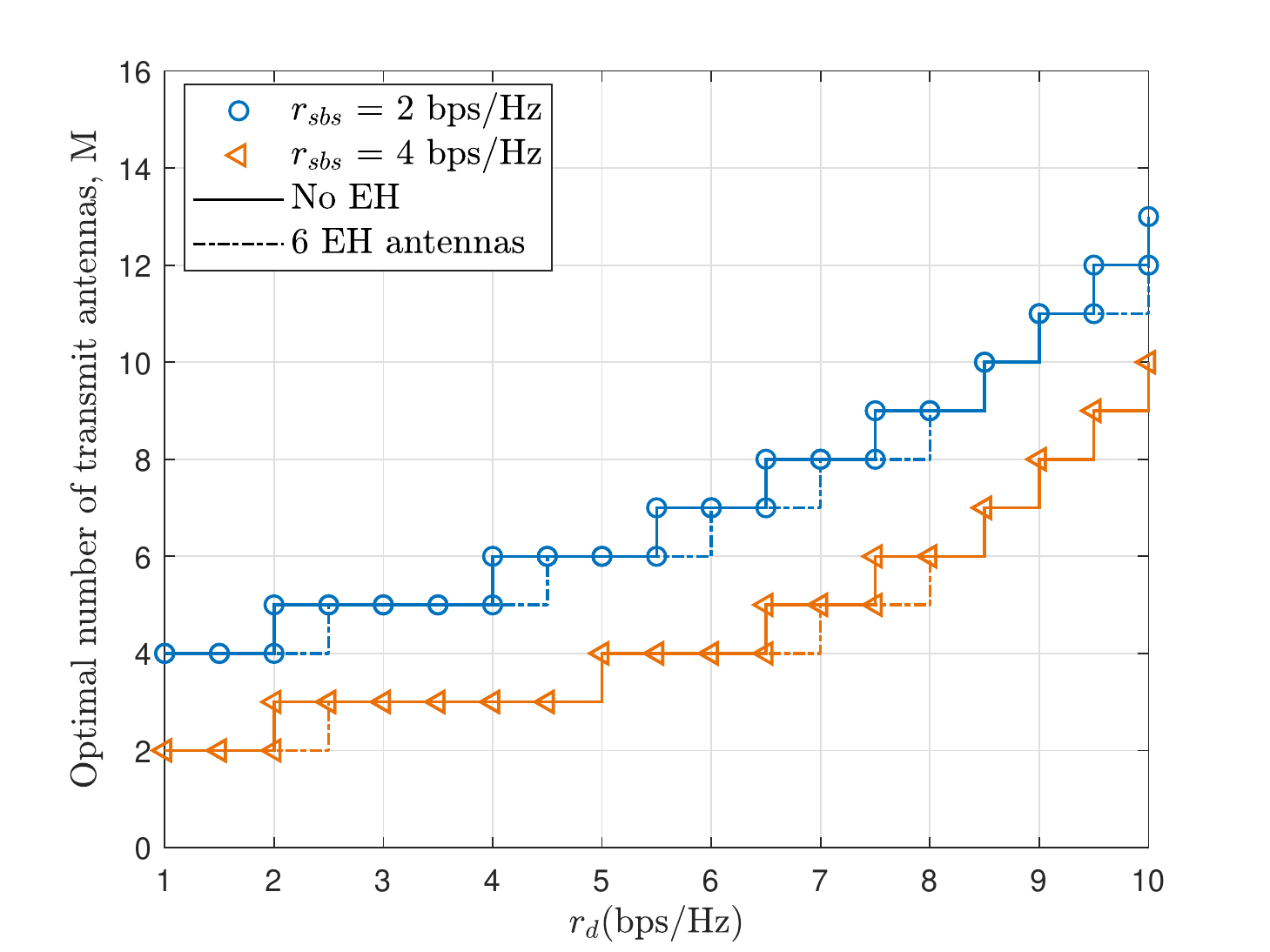}
    \caption{Optimal number of Tx antennas vs $r_d$. $\varphi_\mathrm{ur}=\varphi_\mathrm{td}=$ -80 dB and $Q =$ 16.}
    \label{optimal_tx_antennas}
    \vspace{-4mm}
\end{figure}
\begin{figure}[t!]
    \centering
    \includegraphics[height=2.35in,width=3.3in]{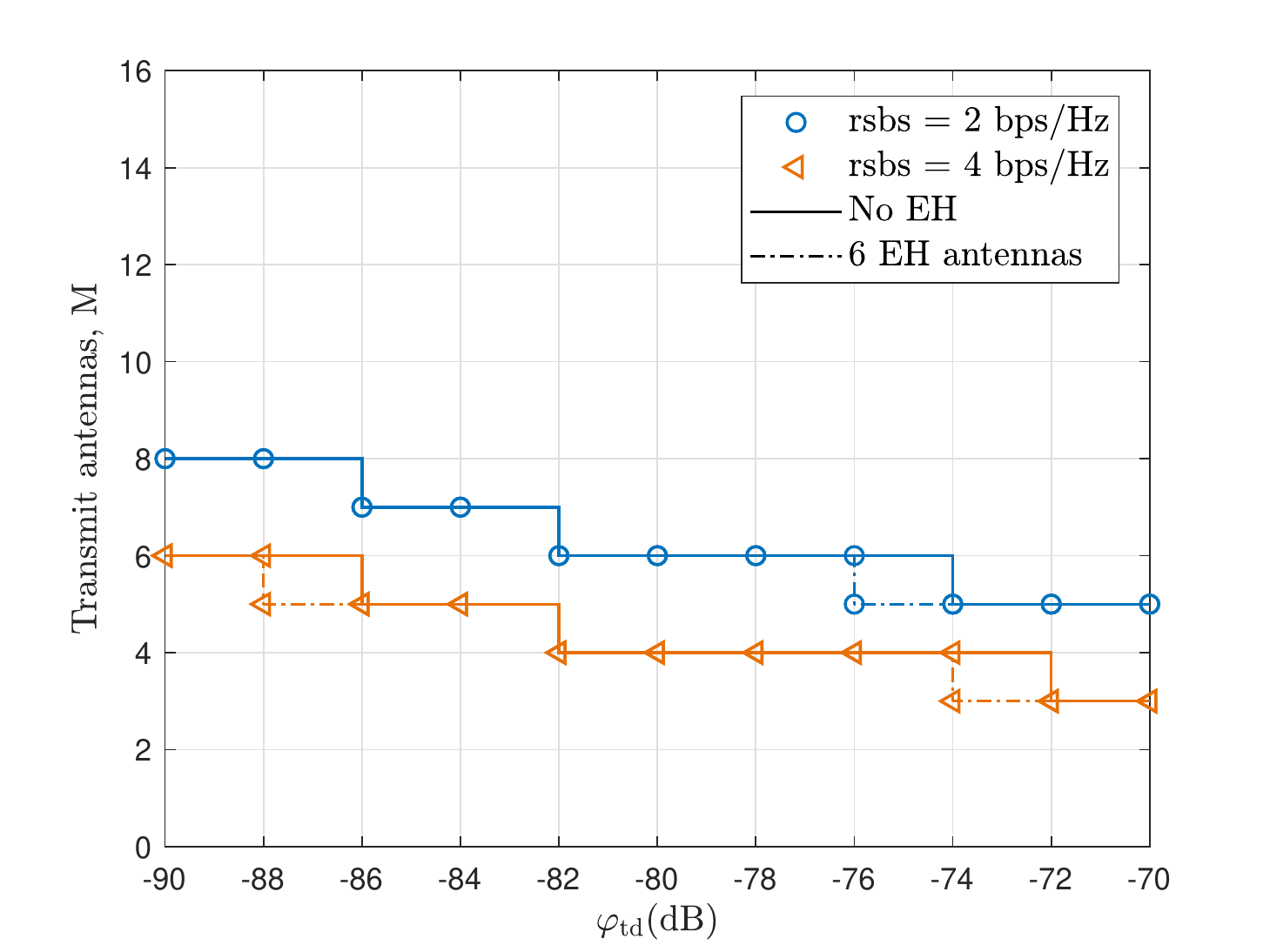}
    \caption{Optimal number of Tx antennas vs path gain. $r_d =$ 4 bps/Hz, $\varphi_\mathrm{ur}=-80$ dB and $Q =$ 16.}
    \label{opt_antennas_vs p_loss}
    \vspace{-4mm}
\end{figure}
\begin{figure}[!htbp]
    \centering
    \includegraphics[height=2.35in,width=3.3in]{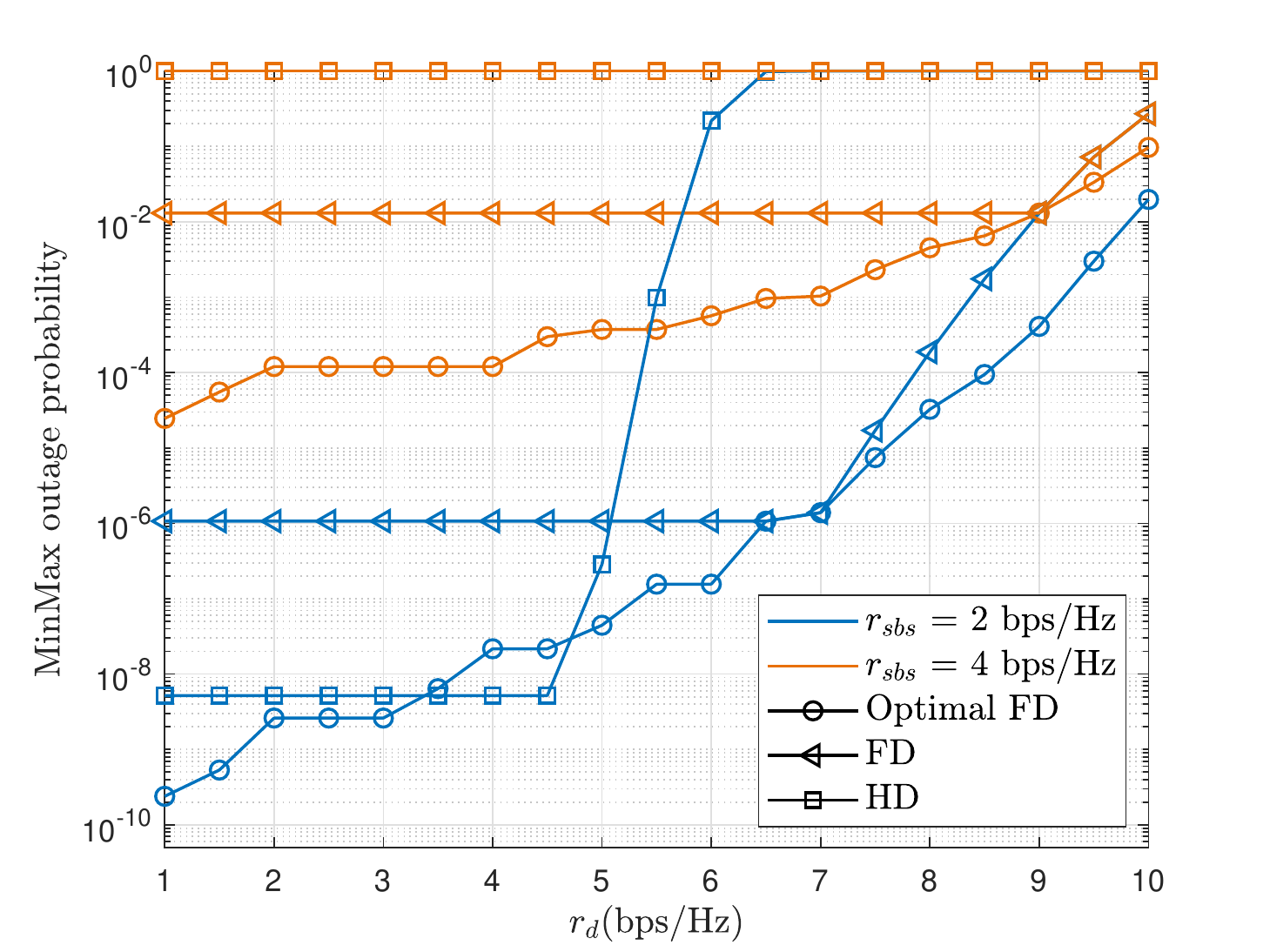}
    \caption{MinMax outage probability. Optimal FD, FD and HD. The area below the dashed line represents the URC region.}
    \label{opt_hd_fd}
    \vspace{-4mm}
\end{figure}
\begin{figure}[t!]
    \centering
    \includegraphics[height=2.35in,width=3.3in]{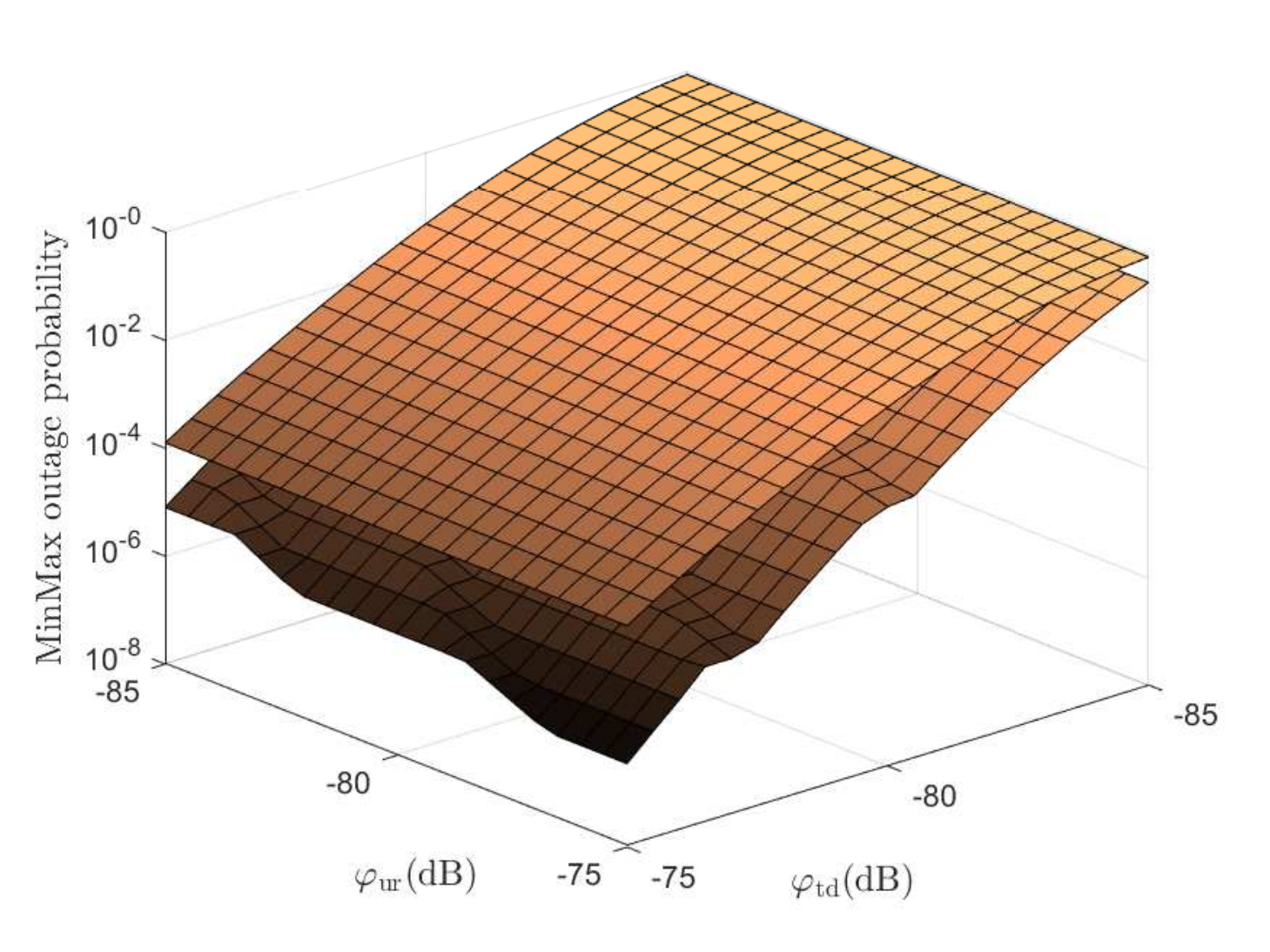}
    \caption{MinMax outage probability in optimal FD (bottom surface) and FD with fixed number of Tx and Rx antennas (top surface), $r_d= r_{sbs} = 6$ bps/Hz and $Q =16$.}
    \label{surface}
    \vspace{-4mm}
\end{figure}
\\\indent Figs. \ref{optimal_tx_antennas} and \ref{opt_antennas_vs p_loss} show the optimal number of transmit antennas  as a function of $r_d$ and path gain, respectively. In Fig. \ref{optimal_tx_antennas} we can notice that the larger the value of $r_d$, the larger the number of transmit antennas required for achieving the optimal performance. On the other hand, Fig. \ref{opt_antennas_vs p_loss} shows how for larger values of path gain, the number of  antennas required for transmission decreases since the SNR at $D$ is higher, thus, the outage probability decreases. This fact allows us to use one of the transmit antennas in reception for improving the outage probability in UR since in the TD link there is a gain in terms of SNR. This way we can reduce the outage in both links.\\
\indent Fig. \ref{opt_hd_fd} compares the performance of a dynamic FD scheme (the SBS determines the optimal configuration dynamically based on the network parameters, which represents the solution to $\mathbf{P1}$), a simple FD with half of the antennas for transmission and half for reception, and an HD schemes using all 16 antennas for both transmission and reception. It shows the performance as a function of $r_d$ and $r_{sbs}$, while we set the path gain to $-$80 dB. We can  notice how the FD scheme outperforms HD's in the whole range of values. For instance, the outage probability is very small in HD when $r_{sbs}$ = 2 bps/Hz and $r_d$ $<$ 4.5 bps/Hz,  but starts increasing when $r_d$ increases since the link TD becomes now the worst. For higher values of $r_{sbs}$ the outage probability even goes to 1 while FD achieves values in the range $10^{-4}-10^{-1}$. The figure also  shows that the optimal FD scheme performs better than the HD case for almost every value of $r_d$. On the other hand, for fixed configuration, FD approaches the optimal one in only one point. For example, the configuration $M = N =$ 8 is optimal when $r_d$ = 7 bps/Hz and $r_{sbs}$ = 2 bps/Hz at a path gain of $-$80 dB. These results evidence that the FD case with dynamic antenna selection significantly improves the system performance when compared to any fixed case.\newline
\indent Fig. \ref{surface} shows the worst outage performance achieved when using the un-optimized and dynamic FD schemes. We compared to the simple FD case because according to Fig. \ref{opt_hd_fd} this one shows better performance than the HD scheme.  Then, the bottom surface represents a lower bound for devices working with similar parameters under other transmission schemes.

\begin{figure}[t!]
    \centering
    \includegraphics[height=2.35in,width=3.3in]{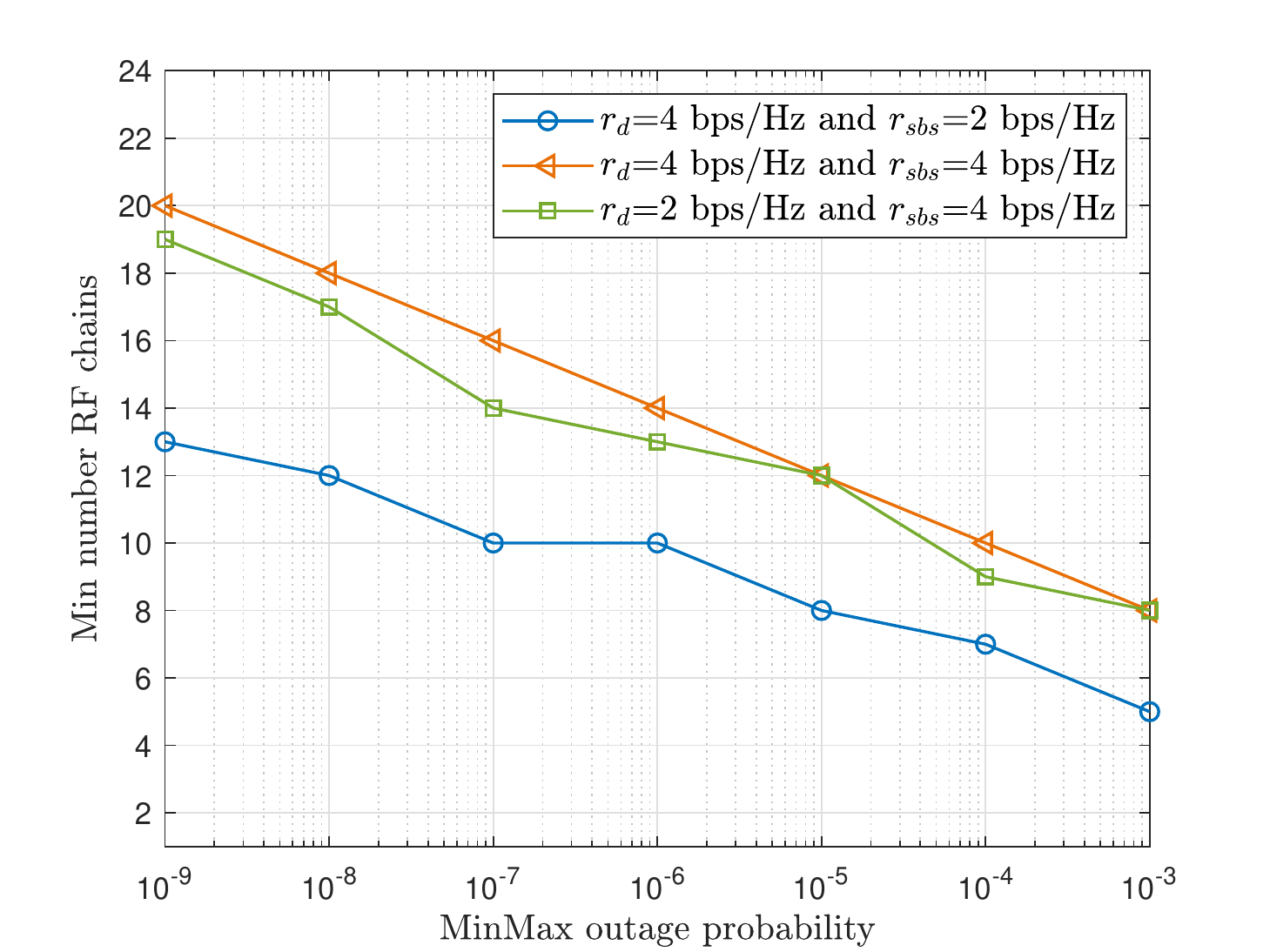}
    \caption{Minimum number of RF chains ($Q$) required for a given MinMax outage probability target. $\varphi_\mathrm{ur}=\varphi_\mathrm{ud}=-80$ dB.}
    \label{opt_rf}
    \vspace{-4mm}
\end{figure}
\begin{figure}
    \centering
    \includegraphics[height=2.35in,width=3.3in]{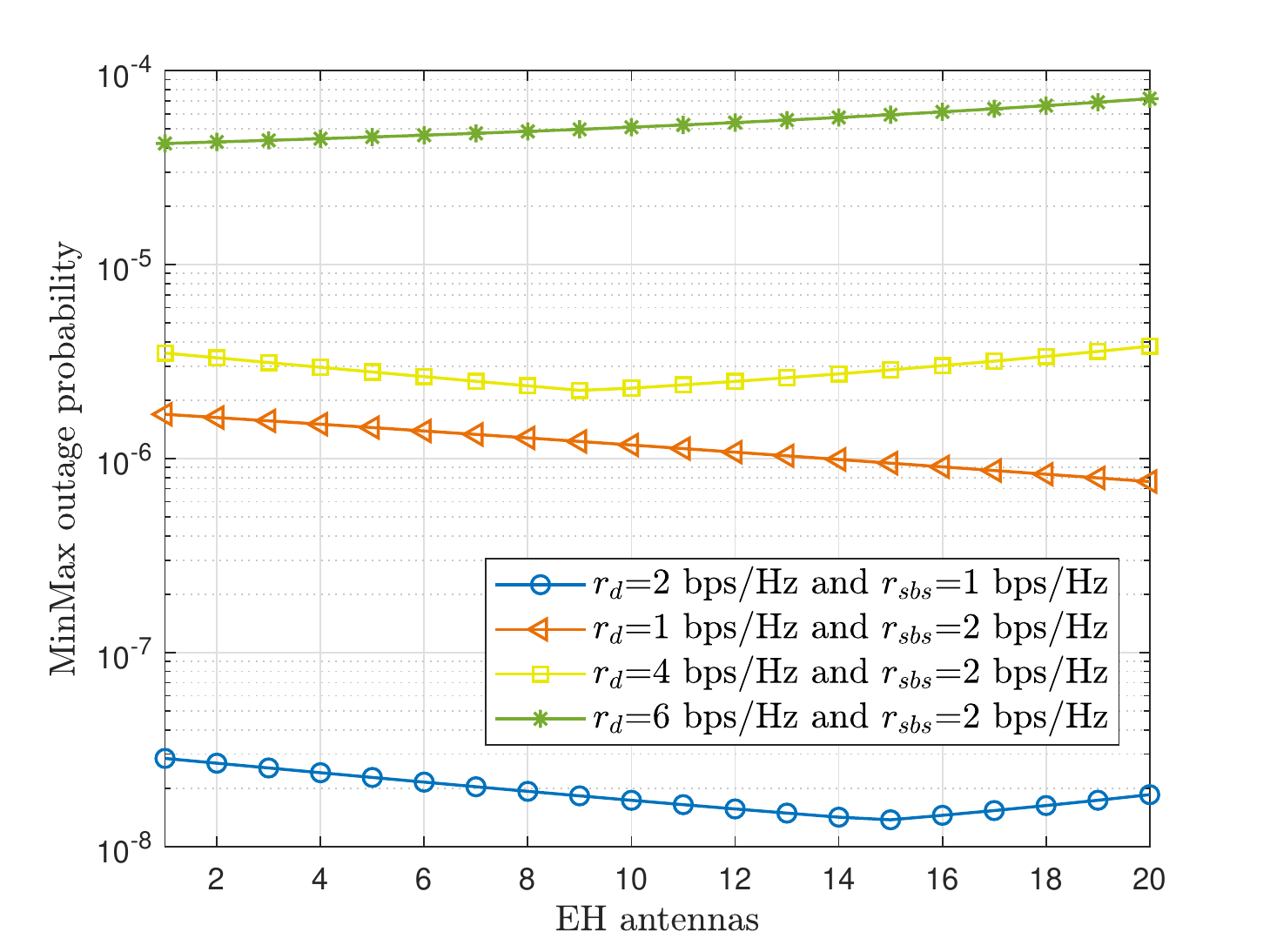}
    \caption{Effect of the sER on the MinMax outage probability. $\varphi_\mathrm{ur}=\varphi_\mathrm{ud}=-80$ dB and $Q =8$.}
    \label{eh_effect}
    \vspace{-4mm}
\end{figure}

\begin{figure*}[t!]
        \centering
        \includegraphics[width = .9\textwidth, height = 3.35in]{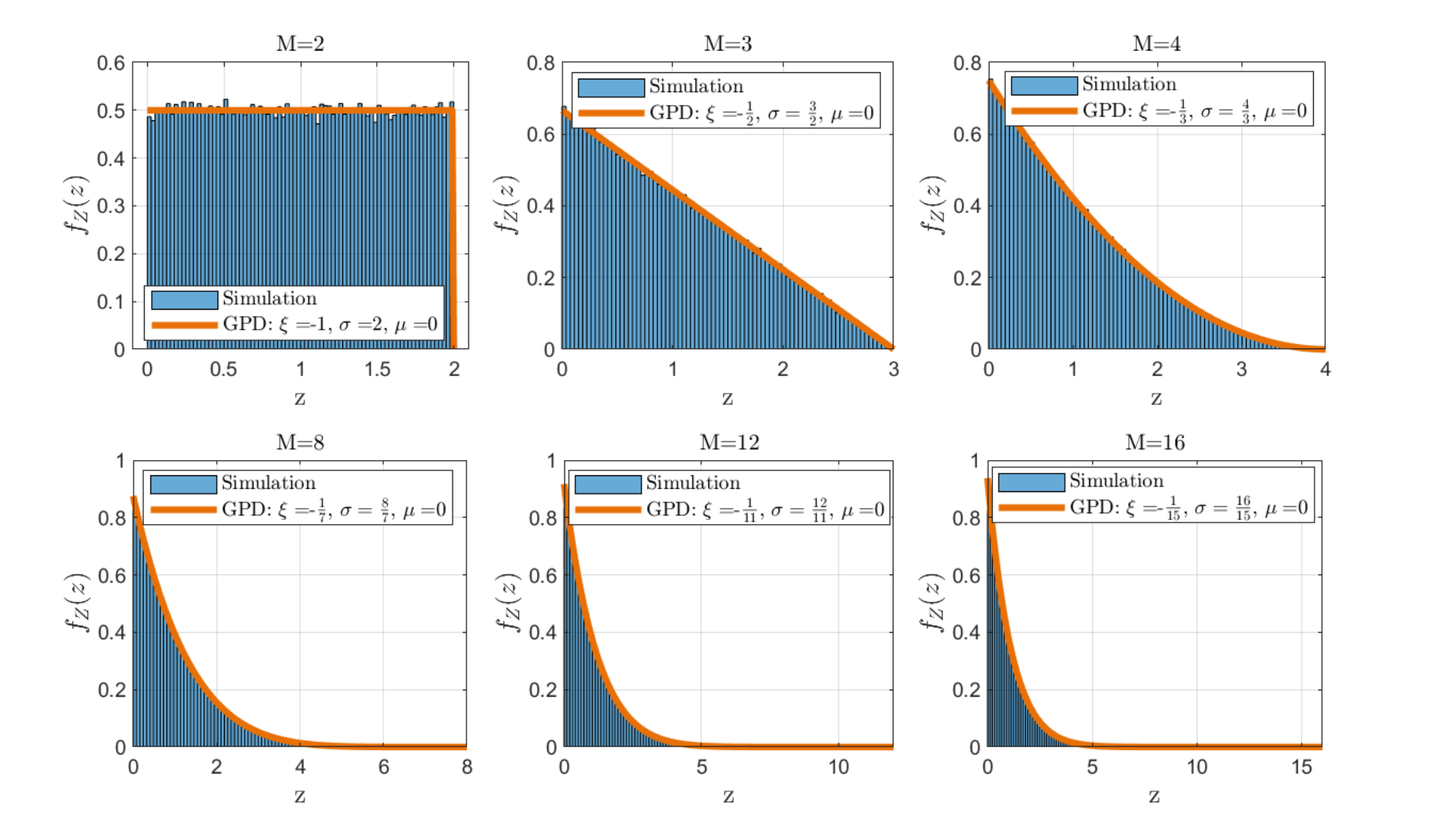}
        \vspace{-4mm}
       \captionsetup{justification=centering}
        \caption{Fitting of $Z$ with the GPD.}
        \vspace{-4mm}
        \label{fig:RV}
\end{figure*}
Fig. \ref{opt_rf} illustrates the results corresponding to the solution to problem $\mathbf{P2}$. Specifically, it shows the minimum number of RF chains required to meet a certain reliability target. Thus, the performance of the worst link will be always better than the minimum outage requirement. It can be noticed that the outage target is guaranteed with a small number of RF chains even for extreme reliability requirements. From these results we observe that the rate at the SBS contributes the most to the minimal number of RF chains, thus, larger $r_{sbs}$ requires larger number of RF chains. Meanwhile,  as $r_d$ increases, the required number of RF chains increases but slowly. In URC, it is expected that users in the UL have lower rates than in the DL and sacrifice 
rate towards reliability. Herein, we show that by doing this not only rate and reliability 
can be guaranteed, but also that the number of required resources  becomes smaller with respect to needed RF chains.\newline
\indent Fig. \ref{eh_effect} shows the effect on the MinMax outage probability of adding more EH antennas at the SBS. As observed, there are variations on the outage values. In general, increasing the number of EH antennas improves the performance in TD but worsens UD. Hence, the sER is beneficial in terms of MinMax outage probability as long as the link TD has the worst performance or has more strict network requirements. We can notice in the figure how the increment in EH antennas for $r_d = 1$ bps/Hz and $r_{sbs} = 2$ bps/Hz is always useful from the outage point of view; however, if the SE in TD increases to 4 bps/Hz and 6 bps/Hz, the curves changes their behaviours. For the former, the link UD becomes the worst from 9 EH antennas on, while for the latter, it is always the worst.\newline
It is worth commenting that the presented results were obtained considering only one UL and one DL devices. Nevertheless, this approach can be extended to communication systems with one SBS and a large number of devices considering the two devices with worst and best outage probabilities for the optimization. Furthermore,  since  we are considering a quasi-static fading scenario, the optimization problems would require to be solved not more than once in each transmission block. This is because when the scenario changes it is only necessary to check if the current configuration still holds and solve the  problems  only  if  it  does  not. Then, the proposed setup becomes more efficient for those scenarios where the fading variations occur slowly, for example, in IoT networks where the sensors are placed at fixed positions.
\vspace{-3mm}
\section{Conclusion}\label{conclusion}
In this work, we analyzed the performance in terms of fair reliability of a scenario where an SBS serves one UL and one DL device under quasi-static Rayleigh fading. We derived expressions for the outage probability for both links under the HD and FD architectures. We also showed the difference in the outage probabilities when considering the power consumed by the RF chains with respect to the ideal case where all the power is radiated. A dynamic FD scheme was proposed, where the SBS can dynamically adjust the number of transmit and receive antennas to attain the best performance. We showed that this FD scheme outperforms the simple FD and HD where a fixed number of antennas is used. We also analyzed the impact of using sER at the SBS, and proved that under some network configurations it varies the optimal number of transmit and receive antennas that achieves the fairest performance of both DL and UL outages probabilities. We also show how a given target reliability can be imposed in the network by finding a minimum number of RF chains that guarantees an optimum operation in both links under URC constraints. This work could be extended to scenarios where there is a larger number of UL and DL devices, and large antenna arrays. Also, future works could include measurements using software defined radios.
\vspace{-4mm}
\appendices
\section{Proof Of Lemma \ref{lemma_1}}
\label{app_A}
Simulation results showed that the RV $Z$ follows a GPD, which is characterized by three parameters: location $\mu$, scale $\sigma$, and shape $\xi$ as 
\begin{align}
f_{\mu,\sigma,\xi}(z)=\frac{1}{\sigma}\bigg(1+\frac{\xi(z-\mu)}{\sigma}\bigg)^{-\frac{1}{\xi}-1}.
\label{eq_GPD}
\end{align}
The distribution parameters are found by curve fitting.  Fig. \ref{fig:RV} shows such fitting as well as the parameters for different numbers of transmit antennas $M$. We realized that these parameters have the following relation with the number of antennas: $\xi=-1/(M-1)$, $\sigma=M/(M-1)$ and $\mu=0$. Hence, substituting these parameters in (\ref{eq_GPD}) we obtain (\ref{pdf}). It can be noticed that the value of this random variable ranges from 0 to $M$ and the expected value is always 1 ($E[z]=\mu-\frac{\sigma}{1-\xi}$). It is also important to highlight that the GPD equals an exponential distribution when both  $\xi$ and $\mu$ are 0, which actually happens when $M\rightarrow\infty$. Although the proof is based on simulations only, it holds since the RV depends only on M.  \hfill 	\qedsymbol
\section{Accuracy and convergence of GL method}
\label{convergence}
	   Fig. \ref{Integral} shows the  convergence of the GL method for sufficiently large orders. The values used for computing the integrals in Fig. \ref{Integral} are shown in Table \ref{table_1}. 
	   \begin{figure}[h!]
	      \centering
	   \vspace{-4mm}
	   \includegraphics[height=2.35in,width=3.3in]{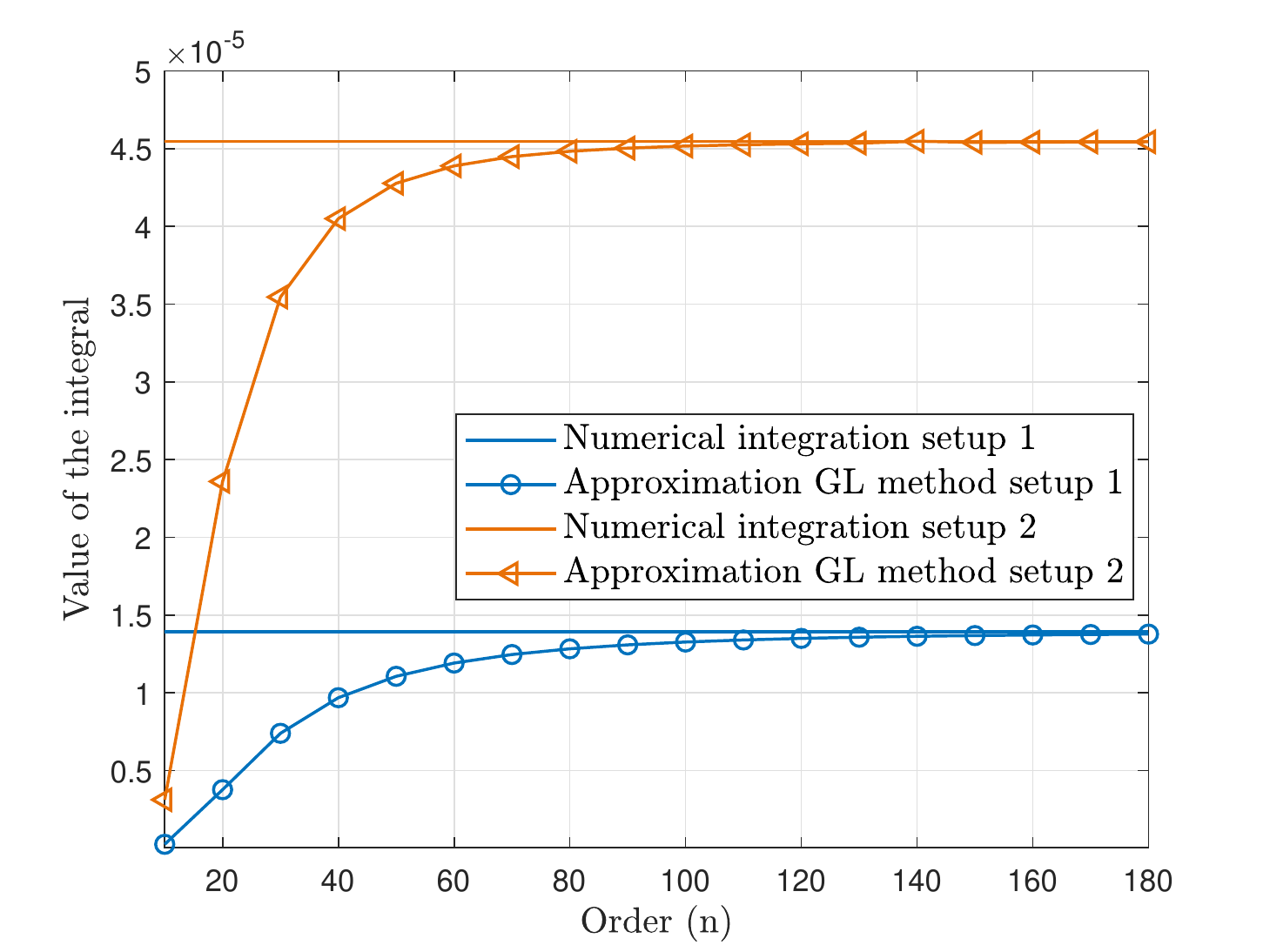}
	      \vspace{-2mm}
	      \caption{Approximation of the integral using Gauss-Laguerre method.}
	      \label{Integral}
	  \end{figure}
\vspace{-2mm}
	   \begin{table}[h!]
\small\addtolength{\tabcolsep}{7pt}
\renewcommand{\arraystretch}{1.5}
\caption{Test values for checking the GL method.}
\label{table_1}
\centering
 \begin{tabular}{|c|c|c|c|c|c|c|}
    \hline
    
    Setup  & $M$  & $N$  & $a_4$ & $l$ &$p$ & $i$\\
    \hline

    1  &  3   &   4  &  20  &  2  &  2  & 2\\
    \hline

    2  &  2   &   2 & 46 & 4 & 4 & 2\\
    \hline
\end{tabular}
\end{table}
\vspace{-2mm}
\bibliographystyle{IEEEtran}
\bibliography{IEEEabrv,References}

\begin{thebibliography}{10}
\providecommand{\url}[1]{#1}
\csname url@samestyle\endcsname
\providecommand{\newblock}{\relax}
\providecommand{\bibinfo}[2]{#2}
\providecommand{\BIBentrySTDinterwordspacing}{\spaceskip=0pt\relax}
\providecommand{\BIBentryALTinterwordstretchfactor}{4}
\providecommand{\BIBentryALTinterwordspacing}{\spaceskip=\fontdimen2\font plus
\BIBentryALTinterwordstretchfactor\fontdimen3\font minus
  \fontdimen4\font\relax}
\providecommand{\BIBforeignlanguage}[2]{{%
\expandafter\ifx\csname l@#1\endcsname\relax
\typeout{** WARNING: IEEEtran.bst: No hyphenation pattern has been}%
\typeout{** loaded for the language `#1'. Using the pattern for}%
\typeout{** the default language instead.}%
\else
\language=\csname l@#1\endcsname
\fi
#2}}
\providecommand{\BIBdecl}{\relax}
\BIBdecl

\bibitem{popovski2019wireless}
P.~Popovski, {\v{C}}.~Stefanovi{\'c}, J.~J. Nielsen, E.~De~Carvalho,
  M.~Angjelichinoski, K.~F. Trillingsgaard, and A.-S. Bana, ``Wireless access
  in {U}ltra-{R}eliable {L}ow-{L}atency {C}ommunication ({URLLC}),'' \emph{IEEE
  Transactions on Communications}, vol.~67, no.~8, pp. 5783--5801, 2019.

\bibitem{6815890}
A.~{Osseiran}, F.~{Boccardi}, V.~{Braun}, K.~{Kusume}, P.~{Marsch},
  M.~{Maternia}, O.~{Queseth}, M.~{Schellmann}, H.~{Schotten}, H.~{Taoka},
  H.~{Tullberg}, M.~A. {Uusitalo}, B.~{Timus}, and M.~{Fallgren}, ``Scenarios
  for 5{G} mobile and wireless communications: the vision of the {METIS}
  project,'' \emph{IEEE Communications Magazine}, vol.~52, no.~5, pp. 26--35,
  May 2014.

\bibitem{6881174}
A.~{Frotzscher}, U.~{Wetzker}, M.~{Bauer}, M.~{Rentschler}, M.~{Beyer},
  S.~{Elspass}, and H.~{Klessig}, ``Requirements and current solutions of
  wireless communication in industrial automation,'' in \emph{IEEE
  International Conference on Communications Workshops (ICC)}, June 2014, pp.
  67--72.

\bibitem{Mahmood.2020}
N.~H. Mahmood, O.~L.~A. L\'opez, O.~S. Park, I.~Moerman, K.~Mikhaylov,
  E.~Mercier, A.~Munari, F.~Clazzer, S.~B\"ocker, and {H. Bartz (Eds.)},
  ``White paper on critical and massive machine type communication towards
  {6G},'' \emph{6G Research Visions}, vol. 2020, no.~11, 2020,
  \url{http://jultika.oulu.fi/files/isbn9789526226781.pdf}.

\bibitem{10}
P.~Kamalinejad, C.~Mahapatra, Z.~Sheng, S.~Mirabbasi, V.~C. Leung, and Y.~L.
  Guan, ``Wireless energy harvesting for the {I}nternet of {T}hings,''
  \emph{IEEE Communications Magazine}, vol.~53, no.~6, pp. 102--108, 2015.

\bibitem{lopez2021massive}
O.~L. L{\'o}pez, H.~Alves, R.~D. Souza, S.~Montejo-S{\'a}nchez, E.~M.~G.
  Fern{\'a}ndez, and M.~Latva-Aho, ``{Massive Wireless Energy Transfer:
  Enabling Sustainable IoT Toward 6G Era},'' \emph{IEEE Internet of Things
  Journal}, vol.~8, no.~11, pp. 8816--8835, 2021.

\bibitem{lopez2017ultrareliable}
O.~L.~A. L{\'o}pez, H.~Alves, R.~D. Souza, and E.~M.~G. Fern{\'a}ndez,
  ``Ultrareliable short-packet communications with wireless energy transfer,''
  \emph{IEEE Signal Processing Letters}, vol.~24, no.~4, pp. 387--391, 2017.

\bibitem{1}
A.~{Sabharwal}, P.~{Schniter}, D.~{Guo}, D.~W. {Bliss}, S.~{Rangarajan}, and
  R.~{Wichman}, ``In-band full-duplex wireless: Challenges and opportunities,''
  \emph{IEEE Journal on Selected Areas in Communications}, vol.~32, no.~9, pp.
  1637--1652, Sep. 2014.

\bibitem{2}
E.~Ahmed and A.~M. Eltawil, ``All-digital self-interference cancellation
  technique for full-duplex systems,'' \emph{IEEE Transactions on Wireless
  Communications}, vol.~14, no.~7, pp. 3519--3532, 2015.

\bibitem{gu2017ultra}
Y.~Gu, H.~Chen, Y.~Li, and B.~Vucetic, ``Ultra-reliable short-packet
  communications: Half-duplex or full-duplex relaying?'' \emph{IEEE Wireless
  Communications Letters}, vol.~7, no.~3, pp. 348--351, 2017.

\bibitem{3}
Y.~{Zeng} and R.~{Zhang}, ``Full-duplex wireless-powered relay with self-energy
  recycling,'' \emph{IEEE Wireless Communications Letters}, vol.~4, no.~2, pp.
  201--204, April 2015.

\bibitem{lopez2020full}
O.~L.~A. L{\'o}pez and H.~Alves, ``Full duplex and wireless-powered
  communications,'' in \emph{Full-Duplex Communications for Future Wireless
  Networks}.\hskip 1em plus 0.5em minus 0.4em\relax Springer, 2020, pp.
  219--248.

\bibitem{4}
W.~{Wu}, B.~{Wang}, Z.~{Deng}, and H.~{Zhang}, ``Secure beamforming for
  full-duplex wireless powered communication systems with self-energy
  recycling,'' \emph{IEEE Wireless Communications Letters}, vol.~6, no.~2, pp.
  146--149, April 2017.

\bibitem{7444859}
{Hongjun Kim}, J.~{Kang}, S.~{Jeong}, K.~E. {Lee}, and J.~{Kang}, ``Secure
  beamforming and self-energy recycling with full-duplex wireless-powered
  relay,'' in \emph{13th IEEE Annual Consumer Communications Networking
  Conference (CCNC)}, Jan 2016, pp. 662--667.

\bibitem{7519055}
Z.~{Hu}, C.~{Yuan}, F.~{Zhu}, and F.~{Gao}, ``Weighted sum transmit power
  minimization for full-duplex system with {SWIPT} and self-energy recycling,''
  \emph{IEEE Access}, vol.~4, pp. 4874--4881, 2016.

\bibitem{13}
A.~{Yadav}, O.~A. {Dobre}, and H.~V. {Poor}, ``Is self-interference in
  full-duplex communications a foe or a friend?'' \emph{IEEE Signal Processing
  Letters}, vol.~25, no.~7, pp. 951--955, July 2018.

\bibitem{11}
A.~A. Nasir, H.~D. Tuan, T.~Q. Duong, and H.~V. Poor, ``Full-duplex {MIMO-OFDM}
  communication with self-energy recycling,'' \emph{arXiv preprint
  arXiv:1903.09931}, 2019.

\bibitem{shaikh2020energy}
M.~H.~N. Shaikh, V.~A. Bohara, P.~Aggarwal, and A.~Srivastava, ``{Energy
  Efficiency Evaluation for Downlink Full-Duplex Nonlinear MU-MIMO-OFDM System
  With Self-Energy Recycling},'' \emph{IEEE Systems Journal}, vol.~14, no.~3,
  pp. 3313--3324, 2020.

\bibitem{nguyen2020transmit}
B.~C. Nguyen, X.~N. Tran \emph{et~al.}, ``Transmit antenna selection for
  full-duplex spatial modulation multiple-input multiple-output system,''
  \emph{IEEE Systems Journal}, vol.~14, no.~4, pp. 4777--4785, 2020.

\bibitem{zhai2018accumulate}
D.~Zhai, H.~Chen, Z.~Lin, Y.~Li, and B.~Vucetic, ``Accumulate then transmit:
  Multiuser scheduling in full-duplex wireless-powered {IoT} systems,''
  \emph{IEEE Internet of Things Journal}, vol.~5, no.~4, pp. 2753--2767, 2018.

\bibitem{zhao2019hybrid}
X.~Zhao, Y.~Zhang, S.~Geng, F.~Du, Z.~Zhou, and L.~Yang, ``Hybrid precoding for
  an adaptive interference decoding {SWIPT} system with full-duplex {IoT}
  devices,'' \emph{IEEE Internet of Things Journal}, vol.~7, no.~2, pp.
  1164--1177, 2019.

\bibitem{guo2019performance}
J.~Guo, S.~Zhang, N.~Zhao, and X.~Wang, ``Performance of {SWIPT} for
  full-duplex relay system with co-channel interference,'' \emph{IEEE
  Transactions on Vehicular Technology}, vol.~69, no.~2, pp. 2311--2315, 2019.

\bibitem{xia2014low}
X.~Xia, K.~Xu, D.~Zhang, and Y.~Xu, ``Low-complexity transceiver design and
  antenna subset selection for cooperative half-and full-duplex relaying
  systems,'' in \emph{IEEE Global Communications Conference}.\hskip 1em plus
  0.5em minus 0.4em\relax IEEE, 2014, pp. 3314--3319.

\bibitem{echevarria_perez_2021}
D.~Echevarr\'ia~P\'erez, ``Reliability performance analysis of half-duplex and
  full-duplex schemes with self-energy recycling,'' Master's thesis, University
  of Oulu, 2021, \url{http://jultika.oulu.fi/files/nbnfioulu-202104017474.pdf}.

\bibitem{culbertson2003full}
D.~L. Culbertson and R.~F. Travelyn, ``Full duplex transceiver,'' Dec.~16 2003,
  uS Patent 6,665,276.

\bibitem{olver2010nist}
F.~W. Olver, D.~W. Lozier, R.~F. Boisvert, and C.~W. Clark, \emph{NIST handbook
  of mathematical functions hardback and CD-ROM}.\hskip 1em plus 0.5em minus
  0.4em\relax Cambridge university press, 2010.

\bibitem{wu2017robust}
W.~Wu, B.~Wang, Y.~Zeng, H.~Zhang, Z.~Yang, and Z.~Deng, ``Robust secure
  beamforming for wireless powered full-duplex systems with self-energy
  recycling,'' \emph{IEEE Transactions on Vehicular Technology}, vol.~66,
  no.~11, pp. 10\,055--10\,069, 2017.

\bibitem{mohammadi2015full}
M.~Mohammadi, H.~A. Suraweera, G.~Zheng, C.~Zhong, and I.~Krikidis,
  ``Full-duplex {MIMO} relaying powered by wireless energy transfer,'' in
  \emph{IEEE 16th International Workshop on Signal Processing Advances in
  Wireless Communications (SPAWC)}.\hskip 1em plus 0.5em minus 0.4em\relax
  IEEE, 2015, pp. 296--300.

\bibitem{lopez2017wireless}
O.~L.~A. L{\'o}pez, E.~M.~G. Fern{\'a}ndez, R.~D. Souza, and H.~Alves,
  ``Wireless powered communications with finite battery and finite
  blocklength,'' \emph{IEEE Transactions on Communications}, vol.~66, no.~4,
  pp. 1803--1816, 2017.

\bibitem{6}
D.~Tse and P.~Viswanath, \emph{Fundamentals of wireless communication}.\hskip
  1em plus 0.5em minus 0.4em\relax Cambridge university press, 2005.

\bibitem{5}
M.~H. Mickle, M.~Mi, and L.~Mats, ``Multiple antenna energy harvesting,'' May~5
  2009, uS Patent 7,528,698.

\bibitem{8}
A.~S. {Arifin} and T.~{Ohtsuki}, ``Outage probability analysis in bidirectional
  full-duplex {SISO} system with self-interference,'' in \emph{The 20th
  Asia-Pacific Conference on Communication (APCC2014)}, Oct 2014, pp. 6--8.

\bibitem{chen2017wireless}
Y.~Chen, N.~Zhao, and M.-S. Alouini, ``Wireless energy harvesting using signals
  from multiple fading channels,'' \emph{IEEE Transactions on Communications},
  vol.~65, no.~11, pp. 5027--5039, 2017.

\bibitem{jeffrey2007table}
A.~Jeffrey and D.~Zwillinger, \emph{Table of integrals, series, and
  products}.\hskip 1em plus 0.5em minus 0.4em\relax Elsevier, 2007.

\bibitem{7}
S.~Cui, A.~J. Goldsmith, and A.~Bahai, ``Energy-efficiency of {MIMO} and
  cooperative {MIMO} techniques in sensor networks,'' \emph{IEEE Journal on
  Selected Areas in Communications}, vol.~22, no.~6, pp. 1089--1098, 2004.

\bibitem{alves2020full}
H.~Alves, T.~Riihonen, and H.~A. Suraweera, \emph{Full-Duplex Communications
  for Future Wireless Networks}.\hskip 1em plus 0.5em minus 0.4em\relax
  Springer, 2020.

\bibitem{zhang2016full}
Z.~Zhang, K.~Long, A.~V. Vasilakos, and L.~Hanzo, ``Full-duplex wireless
  communications: Challenges, solutions, and future research directions,''
  \emph{Proceedings of the IEEE}, vol. 104, no.~7, pp. 1369--1409, 2016.

\bibitem{heino2015recent}
M.~Heino, D.~Korpi, T.~Huusari, E.~Antonio-Rodriguez, S.~Venkatasubramanian,
  T.~Riihonen, L.~Anttila, C.~Icheln, K.~Haneda, R.~Wichman \emph{et~al.},
  ``Recent advances in antenna design and interference cancellation algorithms
  for in-band full duplex relays,'' \emph{IEEE Communications Magazine},
  vol.~53, no.~5, pp. 91--101, 2015.

\bibitem{9}
S.~Nadarajah, ``On the product of generalized pareto random variables,''
  \emph{Applied Economics Letters}, vol.~15, no.~4, pp. 253--259, 2008.

\bibitem{abramowitz1948handbook}
M.~Abramowitz and I.~A. Stegun, \emph{Handbook of mathematical functions with
  formulas, graphs, and mathematical tables}.\hskip 1em plus 0.5em minus
  0.4em\relax US Government printing office, 1948, vol.~55.

\bibitem{agrawal2020noma}
K.~Agrawal, M.~F. Flanagan, and S.~Prakriya, ``{NOMA} with battery-assisted
  energy harvesting full-duplex relay,'' \emph{IEEE Transactions on Vehicular
  Technology}, vol.~69, no.~11, pp. 13\,952--13\,957, 2020.

\end{thebibliography}
\end{document}